\documentclass[12pt]{article}

\usepackage{amssymb}
\usepackage{amsfonts}
\usepackage{amsmath}
\usepackage{amsthm}
\usepackage[english]{babel}
\usepackage{latexsym}
\usepackage{color}
\usepackage{enumerate}
\usepackage{enumitem}
\usepackage{nicefrac}
\usepackage{mathrsfs}
\usepackage{comment}
\usepackage[hmargin=2cm,vmargin=1.9cm]{geometry}
\usepackage{bm}
\usepackage{bbm}
\usepackage{todonotes}
\usepackage{graphicx} 
\usepackage{float} 
\usepackage{hyperref}
\usepackage{soul}

\usepackage{centernot}

\theoremstyle{plain}
\newtheorem{theorem}{Theorem}[section]
\newtheorem{lemma}[theorem]{Lemma}
\newtheorem{proposition}[theorem]{Proposition}
\newtheorem{corollary}[theorem]{Corollary}

\theoremstyle{definition}
\newtheorem{definition}[theorem]{Definition}

\newtheorem*{assumption*}{Standing Assumption}
\newtheorem{remark}[theorem]{Remark}
\newtheorem{example}[theorem]{Example}

\theoremstyle{remark}

\numberwithin{equation}{section}

\newcommand{\E}{\mathbb{E}}

\newcommand{\NN}{\mathbb{N}}
\newcommand{\probp}{\mathbb{P}}

\newcommand{\R}{\mathbb{R}}
\newcommand{\N}{\mathbb{N}}

\renewcommand{\P}{\probp}

\newcommand{\cF}{{\mathcal{F}}}

\newcommand{\cE}{{\mathcal{E}}}

\newcommand{\cS}{{\mathcal{S}}}
\newcommand{\cA}{\mathcal{A}}
\newcommand{\cC}{\mathcal{C}}

\newcommand{\cH}{\mathcal{H}}

\newcommand{\cM}{\mathcal{M}}

\newcommand{\cR}{\mathcal{R}}
\newcommand{\cU}{\mathcal{U}}

\newcommand{\cX}{{\mathcal{X}}}

\newcommand{\ind}{\mathbbm 1}
\newcommand{\esssup}{\mathop {\rm ess\,sup}\nolimits}
\newcommand{\essinf}{\mathop {\rm ess\,inf}\nolimits}

\newcommand{\dom}{\mathop {\rm dom}\nolimits}

\newcommand{\lnorm}{\left\|}
\newcommand{\rnorm}{\right\|}

\newcommand{\VaR}{\mathop {\rm VaR}\nolimits}
\newcommand{\LVaR}{\mathop {\rm LVaR}\nolimits}

\newcommand{\ES}{\mathop {\rm ES}\nolimits}
\newcommand{\OCE}{\mathop {\rm OCE}\nolimits}

\def\keywords{\vspace{.5em}
{\noindent\textbf{Keywords}:\,\relax%
}}

\makeatletter
\def\@fnsymbol#1{\ensuremath{\ifcase#1\or 1\or 2\or 3\or 4\or 5\or 6\or 7\or 8\else\@ctrerr\fi}}
\makeatother

\begin{document}

\title{Risk, utility and sensitivity to large losses\thanks{We thank Felix-Benedikt Liebrich and Ruodu Wang for helpful comments.}}

\author{
Martin Herdegen \\ \textit{University of Warwick}
\and
Nazem Khan \\ \textit{Dublin City University}
\and
Cosimo Munari \\ \textit{University of Verona}
}

\date{\today}

\maketitle

\begin{abstract}
Risk and utility functionals are fundamental building blocks in economics and finance. In this paper we investigate under which conditions a risk or utility functional is sensitive to the accumulation of losses in the sense that any sufficiently large multiple of a position that exposes an agent to future losses has positive risk or negative utility. We call this property sensitivity to large losses and provide necessary and sufficient conditions thereof that are easy to check for a very large class of risk and utility functionals. In particular, our results do not rely on convexity and can therefore also be applied to most examples discussed in the recent literature, including (non-convex) star-shaped risk measures or $S$-shaped utility functions encountered in prospect theory. As expected, Value at Risk generally fails to be sensitive to large losses. More surprisingly, this is also true of Expected Shortfall. By contrast, expected utility functionals as well as (optimized) certainty equivalents are proved to be sensitive to large losses for many standard choices of concave and nonconcave utility functions, including $S$-shaped utility functions. We also show that Value at Risk and Expected Shortfall become sensitive to large losses if they are either properly adjusted or if the property is suitably localized.
\end{abstract}



\parindent 0em \noindent


\section{Introduction}

{\bf Historical perspective}. The problem of defining, measuring, and managing ``risk'' has always occupied a central position in economics and finance. While reviewing the historical developments of different, yet strongly interrelated disciplines like decision theory, portfolio theory or capital adequacy theory, one can find a common thread in the way the notion of ``risk'' has been broadly understood. Indeed, in the pioneering contributions of von Neumann and Morgenstern \cite{VonNeumannMorgenstern1947}, Friedman and Savage \cite{FriedmanSavage1948}, Markowitz \cite{Markowitz1952} and Sharpe \cite{Sharpe1964} as well as in much of the subsequent literature, ``risk'' was related to the entire spectrum of the possible future outcomes of a given economic variable and, more specifically, to the degree of dispersion of its outcomes around its expected value. As such, variance imposed itself as the reference measure of risk in both economics and finance. A powerful connection between the two fields was the recognition that, whenever a utility function is approximated by a quadratic Taylor expansion, risk-averse agents in a classical von Neumann-Morgenstern sense make choices according to a mean-variance criterion; see Markowitz \cite{Markowitz1959}. However, focusing on the entire distribution of future outcomes raises a natural question: Shouldn't economic agents be more concerned with outcomes \emph{below} the expected value? In other words, ``risk'' should be defined and quantified by respecting the natural asymmetry between profits and losses, which clearly have opposite financial meaning and lead to very different economic consequences. In fact, the first contribution in this direction is already \cite{Markowitz1959}, where the notion of semi-variance is put forward as a way to overcome the congenital inability of the variance to discriminate between fluctuations above and below the expected value.

\smallskip

In finance, the seminal article by Artzner, Delbaen, Eber and Heath \cite{artzner1999coherent} triggered an impressive amount of research on the topic of ``risk measures'' with applications to capital adequacy, capital allocation and portfolio selection as well as other areas. Among the class of risk measures defined in that paper, the most prominent ones both from a theoretical and practical perspective are Value at Risk and Expected Shortfall. They share the characteristic of being tail risk measures, i.e., unlike the variance, they depend only on a small portion of the distribution of the future outcomes of the financial variable under study. More precisely, Value at Risk corresponds to the threshold where the tail starts whereas Expected Shortfall is an average of the outcomes in the tail. In particular, Expected Shortfall captures tail risk in a more comprehensive manner. Consequently, in a desire to move towards a more tail-sensitive regime, the Basel III accord changed the official risk measure for market risk in banking regulation from Value at Risk to Expected Shortfall. As put in the Executive Summary of the Consultative Document on the ``Fundamental Review of the Trading Book'' issued by the Basel Committee on Banking Supervision in May 2012 \cite{Basel}:
\begin{quote}
A number of weaknesses have been identified with using value-at-risk (VaR) for determining regulatory capital requirements, including its inability to capture ``tail risk''. For this reason, the Committee has considered alternative risk metrics, in particular expected shortfall (ES). ES measures the riskiness of a position by considering both the size and the likelihood of losses above a certain confidence level. In other words, it is the expected value of those losses
beyond a given confidence level. The Committee recognises that moving to ES could entail certain operational challenges; nonetheless it believes that these are outweighed by the benefits of replacing VaR with a measure that better captures tail risk.
\end{quote}
On a more theoretical level, there have been several attempts to study alternative tail risk measures and describe a general framework for tail risk measures with capital adequacy in mind, see, e.g., \cite{BelliniKlarMullerRosazza2014,BignozziBurzoniMunari2020,BurzoniMunariWang2022,ContDeguestScandolo2010,KouPeng2014,KouPeng2016,LiuWang2021}. In particular, the defining property of a tail risk measure in \cite{KouPeng2016,LiuWang2021} is that it should preserve the ranking of random variables under addition of comonotonic risk or that the risk measure does not discriminate between random variables whose quantiles beyond a certain level coincide. One should also mention that the topic ``tail risk'' has traditionally been central in actuarial science due to the catastrophic nature of many insurance losses. For example the literature on premium principles, dependence theory and extreme value theory is often concerned with the study of suitable tail-related concepts, like that of tail dependence or the loss-tail actuarial index; see, e.g.,  \cite{EmbrechtsKluppelbergMikosch1999,HaanFerreira2006,Wang1996,Wang1998,WangWang2016}.

\smallskip

In decision theory one can identify a similar pattern shifting the focus of research towards a notion of ``risk'' that accounts for the asymmetry between profits and losses. By summarizing and adding new evidence to a number of paradoxes from the previous literature, the seminal paper by Kahneman and Tversky \cite{KahnemanTversky1979} on prospect theory demonstrated that classical expected utility theory fails to capture some fundamental asymmetries in the observed decision behaviour of agents that can be explained by the different roles played by profits and losses. As put in that paper:
\begin{quote}
[...] the reflection of prospects around 0 reverses the preference order. We label this
pattern the reflection effect. [...] note that the reflection effect implies that risk aversion in the positive domain is accompanied by risk
seeking in the negative domain. [...] our data are incompatible with the notion that certainty is generally desirable. Rather, it appears that certainty increases the aversion of losses as well as the desirability of gains.
\end{quote}
These considerations shifted the focus of much of the subsequent research from risk aversion to loss aversion; see \cite{KahnemanKnetschThaler1991,RabinThaler2001,TverskyKahneman1991} and the references therein. Some of the key experiments of prospect theory support diminishing sensitivity for losses, i.e., the marginal value of losses decreases with their size, and several indices of loss aversion have been proposed in the literature; see, e.g., \cite{SchmidtTraub2002,TverskyKahneman1991}. More recently, the risk seeking behaviour for losses advocated by prospect theory has been challenged and a more complicated co-existence of risk seeking and risk aversion, especially in the presence of large losses, has been suggested and tested empirically; see, e.g., \cite{bosch2006reflections,bosch2010averting}. For our purposes, it is also relevant to highlight \cite{giesecke2008measuring}, which lies at the intersection of risk measure theory and decision theory and focuses on comparing standard regulatory risk measures like Value at Risk and Expected Shortfall with utility-based risk measures. The latter quantify risk in terms of the smallest amount of capital that has to be injected into an outstanding financial position in order to make the modified position preferable to a benchmark position from the point of view of classical expected utility theory. In the context of a case study, the authors argue that a utility-based risk measure, for instance the one induced by power utility, may be more ``sensitive with respect to extreme events'' compared to both Value at Risk and Expected Shortfall. In that paper sensitivity to extreme events is captured by the dependence of a risk measure on a parameter that represents the conditional expected loss of an extreme event. Another relevant recent paper is \cite{ArmstrongBrigo2019}, where it is argued that a utility-based risk measure does a better job than risk measures like Value at Risk and Expected Shortfall in constraining excessive loss taking behavior of portfolio investors seeking to maximize an S-shaped utility like that postulated by prospect theory.

\smallskip

{\bf Contribution of the paper}. Our aim in this paper is to revisit the topic of ``loss sensitivity'' from first principles. More precisely, we are interested in studying under which conditions a risk or utility functional is sensitive to the accumulation of large losses in the sense that any sufficiently large multiple of a position that exposes the agent to downside risk moves the scaled position from acceptable to unacceptable risk or from utility to disutility. We call this property  {\em sensitivity to large losses}. Failing to meet this condition means using a rule to quantify risk or utility that, in the worst case, is blind to potentially catastrophic losses and, in the best case, allows gains to compensate for large losses, which, as argued below, is at least problematic in certain financial applications. After defining sensitivity to large losses in a rigorous way, we move on to establishing a number of sufficient and necessary conditions for a risk or utility functional to satisfy this property. In particular, we aim to find tractable conditions that can be applied to a large variety of concrete examples that are relevant from a practical or a theoretical perspective.

\smallskip

The concept of a risk functional, respectively utility functional, is broad and covers all functionals
\[
\cR:\cX\to(-\infty,\infty], \ \ \ \ \cU:\cX\to[-\infty,\infty),
\]
defined on a suitable domain $\cX$ of random variables over a probability space $(\Omega,\cF,\probp)$. Here, the elements of $\cX$ represent net positions of economic agents such as the net wealth of a consumer unit, the return on an investor's portfolio or the equity value of a company. We assume that $\cR$ and $\cU$ are monotone and normalized in the sense that a position with a worse loss profile has higher risk/lower utility and the risk/utility of the null position is set to zero by convention. As a result, we cover a very wide spectrum of concrete functionals ranging from (convex and non-convex) risk measures over (concave and $S$-shaped) expected utility functionals to (optimized) certainty equivalent functionals. To each risk/utility functional as above one can clearly associate a utility/risk functional through
\[
\cU_\cR(X):=-\cR(X), \ \ \ \ \cR_\cU(X):=-\cU(X),
\]
showing that the mathematical treatment of these two types of functionals is symmetrical. For this reason, for the remainder of the introduction and in the statement and proofs of our results, we focus on risk functionals and leave to the reader the equivalent formulations in terms of utility functionals. Nevertheless, as the economic interpretation is not symmetric, we provide definitions, discussions and relevant examples for both types of functionals.

\smallskip

The point of departure for our work is the property of {\em sensitivity to losses}, which stipulates that a risk functional marks as “unacceptable” any position that includes some loss potential by assigning to it a strictly positive risk indicator. Formally, this translates into the requirement:
\[
X\in\cX, \ \probp(X<0)>0 \ \implies \ \cR(X)>0.
\]
In the language of capital adequacy, this is equivalent to saying that a company needs to be always solvent to be deemed adequately capitalized. Clearly, this property is too stringent to make sense from a practical perspective. This is because ensuring solvency with certainty would be extremely costly and would reverberate into correspondingly costly -- most likely prohibitively costly -- commercial conditions for the company's customers, including high maintenance costs and fees, high insurance premia or low interest rates. As a result, a tradeoff between costs and default risk will have to be made, leading customers and regulators on their behalf to accept some default risk in exchange for more attractive commercial conditions.

\smallskip

An economically more compelling way to define “sensitivity to losses” is therefore to impose that positions with {\em sufficiently large} loss peaks are deemed unacceptable. Formally speaking, a risk functional is sensitive to large losses if it marks as “unacceptable” any sufficiently large multiple of a position that contains some losses:
\[
X\in\cX, \ \probp(X<0)>0 \ \implies \ \cR(\lambda X)>0 \ \mbox{for $\lambda>0$ large enough}.
\]
In other words, a risk functional is declared to be insensitive if either it is blind to potentially catastrophic losses or permits gains to compensate for excessively large losses. In the latter case, the problem is twofold. On the one hand, gains and losses are realized in {\em different} future scenarios and it is thus debatable, at a conceptual level, whether a risk functional that mixes gains and losses is a sound indicator of (tail) risk; see, e.g., \cite{cont2013loss,Staum2013}. On the other hand, there are certain applications where the compensation of losses by means of gains creates possible conflicts of interest undermining the very objectives of risk control. For example, if a risk functional is used to determine capital requirements of financial institutions, whose primary goal is to protect liability holders from the risk of default, then it is problematic if a company with a worse default profile is deemed adequately capitalized because of a larger surplus, of which liability holders have no direct share; see, e.g., the discussion on surplus-invariant risk measures in \cite{HePeng2018,KochMorenoMunari2015,GaoMunari2020}.

\smallskip

We begin our analysis by noting the simple but important result that a positively homogeneous risk functional is sensitive to large losses precisely when it is sensitive to losses (Proposition~\ref{prop: positively homogeneous}). This is because, under positive homogeneity, a risk functional scales linearly with the position's size and therefore, in a sense, treats large and small losses alike when it comes to discriminating positions based on the sign of their outcomes. To the best of our knowledge, this important (and somewhat problematic) implication of positive homogeneity has never been fully recognised in the literature. 

\smallskip

The (economically undesirable) equivalence between sensitivity to large losses and sensitivity to losses ceases to hold as soon as one moves away from positive homogeneity. In the inhomogeneous case, we investigate sensitivity to large losses in the class of star-shaped risk functionals, which is much larger than but contains the class of convex risk functionals. For a star-shaped risk functional we prove that sensitivity to large losses holds whenever the risk of any position that contains some losses explodes as a function of the position's size (Theorem~\ref{theo: recession functional}):\begin{equation}
\label{eq: divergence definition}
X\in\cX, \ \probp(X<0)>0 \ \implies \ \cR(\lambda X)\to\infty \ \mbox{as $\lambda\to\infty$}.
\end{equation}
Equivalently, sensitivity to large losses holds whenever, for each position exposing the agent to some losses, {\em some} rescaling of the position is unacceptable (Theorem~\ref{theo: recession functional}): 
\[
X\in\cX, \ \probp(X<0)>0 \ \implies \ \cR(\lambda X)>0 \ \mbox{for some $\lambda>0$}.
\]
We also show that the equivalence generally fails in the absence of star-shapedness. This is because, in general, the sign of a risk functional as a function of the position's size may, somewhat disturbingly, turn from negative to positive and vice versa. Our result therefore provides a new argument in favour of star-shapedness that adds to the recent analysis of this property in \cite{CastagnoliCattelanMaccheroniTebaldiWang2022}. The previous characterizations greatly simplify the verification of sensitivity to large losses under star-shapedness and provide a characterization of this property in terms of the recession functional $\cR^\infty$, i.e., the best positively homogeneous approximation of a risk functional by excess. Most notably, in the class of star-shaped risk functionals that are also cash-additive,
sensitivity to large losses boils down to the recession functional reducing to the worst-case risk measure (Theorem~\ref{theo: cash-additive}):
\[
X\in\cX \ \implies \ \cR^\infty(X):=\sup_{\lambda>0}\frac{\cR(\lambda X)}{\lambda}=\esssup(-X).
\]
Here, the essential supremum $\esssup$ refers to the largest outcome of the random variable to which it is applied. This characterization can be exploited to verify sensitivity to large losses for concrete examples of risk functionals. If star-shapenedness is replaced with the stronger property of convexity, we obtain an additional characterization with an appealing economic interpretation. In this case, a risk functional is sensitive to large losses precisely when it is not possible to use a position carrying downside risk to reduce risk through aggregation. In other words, under sensitivity to large losses, adding a position with downside risk to an outstanding position must increase risk at least in some cases (Theorem~\ref{theo: convex}):
\[
X\in\cX, \ \probp(X<0)>0 \ \implies \ \mbox{$\cR(Y+X)>\cR(Y)$ for some $Y\in\cX$}.
\]

\smallskip

It should be noted that two key risk functionals in practice and the most prominent in capital adequacy, namely Value at Risk and Expected Shortfall, fail to be sensitive to large losses on their natural domain of definition (Example \ref{ex: var and es}). This failure can be viewed as a consequence of their positive homogeneity. While this is interesting {\em per se} and contributes to a better understanding of their merits and limitations, this does not mean that Value at Risk and Expected Shortfall cannot be sensitive to large losses if they are restricted to a suitable domain of financial positions. This observation suggests localizing sensitivity to large losses to a subset of the model space:
\[
X\in\cS\subset\cX, \ \probp(X<0)>0 \ \implies \ \cR(\lambda X)>0 \ \mbox{for $\lambda>0$ large enough}.
\]
It turns out that the results mentioned above remain true if the larger domain $\cX$ is replaced by $\cS$. The advantage of this added flexibility is that the theory becomes much richer and inclusive as more risk functionals including Value at Risk and Expected Shortfall can be shown to be sensitive to large losses \emph{on specific domains of positions}. We mainly focus on three domains in this paper: (1) $\cS$ consists of all positions in $\cX$ that are strictly negative almost surely, in which case the associated sensitivity to large losses is named {\em sensitivity to large sure losses}; (2) $\cS$ consists of all nonpositive positions in $\cX$, in which case we speak of {\em sensitivity to large pure losses}; (3) $\cS$ consists of all positions in $\cX$ with nonpositive expectation, in which case we speak of {\em sensitivity to large expected losses}. It turns out that Value at Risk is sensitive to large sure losses but not to large pure/expected losses, whereas Expected Shortfall is sensitive to large sure/pure/expected losses (Example \ref{ex: var}(a) and Example \ref{ex: ES}).

\smallskip

By requiring only minimal properties on risk and utility functionals, our approach covers both convex/concave and nonconvex/nonconcave functionals. Notably, we can apply our results to nonconvex examples of risk measures such as Value at Risk and some of its variants, as well as to expected utility and (optimized) certainty equivalent functionals associated with nonconcave utility functions like the concave-convex-concave functions in \cite{FriedmanSavage1948} and their modifications in \cite{Markowitz1952b}, the star-shaped functions in \cite{LandsbergerMeilijson1990} or the S-shaped functions from prospect theory \cite{KahnemanTversky1979}. In particular, expected utility is sensitive to large losses whenever the loss of utility in the presence of large losses dominates at the limit the gain in utility in the presence of large gains, which holds for most (but not all) classical (von Neumann-Morgenstern) utility functions but is also satisfied by many utility functions used in prospect theory (Theorem \ref{thm:expect utility}). In our analysis we devote special attention to star-shaped risk and utility functionals, which turn out to be particularly flexible when it comes to studying their sensitivity to large losses. These results can be viewed as a complement to the recent study of star-shaped risk measures in \cite{CastagnoliCattelanMaccheroniTebaldiWang2022}. Among our examples, we treat risk measures based on star-shaped loss functions, a natural and interesting example of a star-shaped risk measure that adds to the list found in that paper.

\smallskip

To the best of our knowledge, the property (but neither the terminology nor the interpretation) of sensitivity to large losses was first used in \cite[Remark 2.7]{ChernyKupper2007} in the context of a special class of concave utility functionals, termed divergence utilities, in the equivalent form~\eqref{eq: divergence definition}. Interestingly, it is said there that the property ``shows an essential difference between divergence utilities and coherent utilities and, in our opinion, could be quite a desirable feature for potential applications of divergence utilities to measuring preferences''. It also occurs in \cite{CheriditoHorstKupperPirvu2016}, where the authors deal with equilibrium pricing and model agents' preferences by means of concave utility functionals that, among other properties, are sometimes required to satisfy sensitivity to large losses in the form~\eqref{eq: divergence definition}. Notwithstanding, the property only plays a technical role and is not further discussed. More recently, sensitivity to large expected losses was used in \cite{HerdegenKhan2024}  to prove the existence of optimal portfolios in a mean-risk framework extending the classical mean-variance Markowitz's setting. \\
On a side note, it is worth pointing out that the strategy of using recession functionals in our study of sensitivity to large losses for risk/utility functionals that are not positively homogeneous has been recently adopted in a number of papers focusing on different topics, ranging from pricing to portfolio selection and capital allocation, but with the common denominator of dealing with ``nonconic'' problems; see \cite{ArducaMunari2023,BaesKochMunari2020,HerdegenKhan2024,Liebrich2024,LiebrichMunari2022,LiebrichSvindland2019,Pennanen2011}. 

\smallskip

The remainder of this paper is organized as follows. In Section~\ref{sec: risk measures} we formally introduce risk and utility functionals and review some of their key properties used in the paper. In Section~\ref{sec: sensitivity to large losses} we define sensitivity to large losses and establish characterizations under positive homogeneity, star-shapedness, convexity, and cash additivity. Moreover, we study localized versions that have partly been considered in the extant literature and explore the relationship with sensitivity to loss concentrations. In Sections~\ref{sect: examples} and \ref{sec:examples utility} we illustrate how our results can be used to verify sensitivity to large losses for various examples of utility and risk functionals. Section \ref{sec:conclusion} concludes. All proofs and some additional results of a more technical nature are delegated to Appendix \ref{sec:appendix}.


\section{Risk and utility functionals}
\label{sec: risk measures}

In this preliminary section, we introduce our model space and list the properties of risk and utility functionals that are used in the paper. We consider a one-period economy where uncertainty about the terminal state of the world is modelled by a probability space $(\Omega,\cF,\probp)$. We denote by $L^1$ the space of integrable random variables modulo almost-sure equality under $\probp$. As usual, we identify numbers in $\R$ with constant random variables. The elements of $L^1$ represent net positions of financial agents at the terminal date, with the convention that a positive outcome of a random variable is interpreted as a profit. The expectation and essential supremum of $X\in L^1$ are denoted by $\E[X]$ and $\esssup(X)$. The positive part and the negative part of $X\in L^0$ are defined, respectively, by $X^+:=\max\{X,0\}$ and $X^-:=\max\{-X,0\}$.

\smallskip

Throughout the paper we fix a topological vector space $\cX\subset L^1$ of reference financial positions, which is partially ordered by the usual almost-sure inequality under $\probp$. The corresponding convex cone of nonnegative positions is denoted by $\cX_+$. We always assume that $\cX$ contains all bounded random variables and is closed under multiplication by indicator functions. Additional requirements on $\cX$ will be added when needed. We say that a functional $\cH:\cX\to[-\infty, \infty]$ is

\begin{itemize}
    \item \emph{normalized} if $\cH(0)=0$.
   \item \emph{decreasing} if $\cH(X)\ge \cH(Y)$ for all $X,Y\in\cX$ with $X\le Y$.
 \item \emph{increasing} if $\cH(X)\le \cH(Y)$ for all $X,Y\in\cX$ with $X\le Y$.
    \item \emph{positively star-shaped} if $\cH(\lambda X)\ge\lambda \cH(X)$ for all $X\in\cX$ and $\lambda\in(1,\infty)$.
        \item \emph{negatively star-shaped} if $\cH(\lambda X)\le\lambda \cH(X)$ for all $X\in\cX$ and $\lambda\in(1,\infty)$.
\end{itemize}
The {\em (effective) domain} of $\cH$ is $\dom(\cH):=\{X\in\cX \,; \ \cH(X)\in\R\}$. Note that what we called positively star-shaped is usually only called star-shaped in the risk measure literature. Alternatively, a positively star-shaped functional is called star-shaped at $0$ and a negatively star-shaped functional is called star-shaped at $\infty$.

\begin{definition}
A {\em risk functional} is any functional $\cR:\cX\to(-\infty,\infty]$ that is normalized and decreasing. The set $\cA_\cR:=\{X\in\cX \,; \ \cR(X)\le0\}$ is called {\em acceptance set} and its elements {\em acceptable positions}. A risk functional that is positively star-shaped is called a \emph{star-shaped risk functional}. A {\em utility functional} is any functional $\cU:\cX\to[-\infty,\infty)$ that is normalized and increasing. The set $\cA_\cU:=\{X\in\cX \,; \ \cU(X)\ge0\}$ is called {\em acceptance set} and its elements {\em acceptable positions}. A utility functional that is negatively star-shaped is called a \emph{star-shaped utility functional}.
\end{definition}

\smallskip

Both defining requirements of a risk/utility functional are economically plausible and satisfied by the vast majority of measures of risk and indicators of utility encountered in theory and practice. Incidentally, note that the normalization property is always satisfied up to a convenient translation by a constant (provided the risk/utility of the null position is finite). This choice allows us to interpret every acceptable position as a position whose risk/utility is not larger/smaller than the risk/utility of the null position. Note also that, by definition, every nonnegative position is acceptable and adding a nonnegative position to an acceptable position leads to another acceptable position under any risk/utility functional. In passing, note that there is no ambiguity as to our definition of acceptability because the only functional that is both a risk and utility functional is the zero functional.

\smallskip

It is clear that, while the economic interpretation of a risk functional and a utility functional is different, as becomes particularly clear when we look at concrete examples, from a mathematical perspective a theory of (star-shaped) risk functionals is completely symmetric to a theory of (star-shaped) utility functionals. This is because it suffices to change sign to switch from a risk to a utility functional and viceversa. As different readers may be interested in different interpretations and results, we formulate the key properties for both classes of functionals but, for efficiency, we have decided to state and prove our results for risk functionals only. In Sections \ref{sect: examples} and \ref{sec:examples utility} we work with concrete examples and apply our results to risk and utility functionals, respectively.

\smallskip

In the sequel we freely use the following properties. A risk functional $\cR:\cX\to(-\infty, \infty]$ is called
\begin{itemize}
    \item \emph{positively homogeneous} if $\cR(\lambda X)=\lambda \cR(X)$ for all $X\in\cX$ and $\lambda\in(0,\infty)$.
    \item \emph{convex} if $\cR(\lambda X+(1-\lambda)Y)\le\lambda \cR(X)+(1-\lambda)\cR(Y)$ for all $X,Y\in\cX$ and $\lambda\in(0,1)$.
    \item \emph{cash-additive} if $\cR(X+m)=\cR(X)-m$ for all $X\in\cX$ and $m\in\R$.
    \item \emph{lower semicontinuous} if, for all nets $(X_k)\subset\cX$ and $X\in\cX$ with $X_k\to X$,
\[
\cR(X) \le \liminf_{k}\cR(X_k).
\]
    \item \emph{lower regular} if, for every $X\in\cX$ and all sequences $(m_n)\subset\R$ and $m\in\R$ with $m_n\to m$,
\[
\cR(X+m) \le \liminf_{n\to\infty} \cR(X+m_n).
\]
\end{itemize}
When considering a utility functional, the properties of convexity and lower semicontinuity/regular- ity translate into concavity and upper semicontinuity/regularity and the property of cash-additivity retains the same name but constants are ``taken out'' of the functional without changing sign. Once again, no confusion can arise because there is no cash-additive risk functional that is also a utility functional.

\begin{remark}
The above properties, especially convexity and positive homogeneity, are standard in the theory of risk/utility functionals. We refer to \cite{FoellmerSchied2016} for a general treatment of cash-additive risk functionals and to \cite{CastagnoliCattelanMaccheroniTebaldiWang2022} for a recent presentation of cash-additive star-shaped risk functionals. For a general treatment beyond cash-additivity, see, e.g., \cite{FarkasKochMunari2014}. While cash-additivity is less mainstream in utility theory, there are two notable exceptions corresponding to numerical representations of preferences in the dual theory of choice in \cite{Yaari1987} and to variational preferences in \cite{MaccheroniMarinacciRustichini2006}. In addition, we refer to \cite{Delbaen2012}. It is worth pointing out the following implications holding for a risk functional:
\begin{enumerate}
    \item[(a)] Both positive homogeneity and convexity imply star-shapedness. 
    \item[(b)] Both lower semicontinuity and cash-additivity imply lower regularity.
\end{enumerate}
\end{remark}

\smallskip

We end this section by recalling the most important examples of risk and utility functionals that we will be using throughout the paper.

\begin{example}
\label{ex: VaR, ES, Eu}
(a) For $X \in L^1$ and $\alpha\in(0,1)$, we define the {\em Value at Risk} of $X$ at level $\alpha$ by
\[
\VaR_\alpha(X) := \inf\{m\in\R \,; \ \probp(X+m<0)\le\alpha\}.
\]
Note that, up to a sign, $\VaR_\alpha(X)$ coincides with the upper $\alpha$-quantile of $X$. The functional $\VaR_\alpha$ quantifies risk as the minimal amount of capital that has to be added to a position to constraint the probability of loss/default below $\alpha$. It is clear that $\VaR_\alpha$ is a cash-additive and positively homogeneous risk functional according to our general definition.

\smallskip

(b) For $X \in L^1$ and $\alpha\in(0,1)$, we define the {\em Expected Shortfall} of $X$ at level $\alpha$ by
\[
\ES_\alpha(X) := \frac{1}{\alpha}\int_0^\alpha\VaR_\beta(X)d\beta.
\]
The functional $\ES_\alpha$ thus quantifies risk as the minimal amount of capital that has to be added to a position to ensure profit/solvency on average in the worst $100\alpha\%$ cases. It is clear that  $\ES_\alpha$ is a cash-additive and positively homogeneous risk functional according to our general definition, which is also known to be convex.

\smallskip

(c) A function $u:\mathbb{R} \to [-\infty,\infty)$ is a \emph{utility function} if it is increasing, normalised, i.e., $u(0) = 0$, and not superlinear at $\infty$ in the sense that
\begin{equation}
\label{eq:def:utility fn:not superlinear}
\limsup_{x \to \infty} \tfrac{u(x)}{x} < \infty.
\end{equation}  
For a utility function $u$, the \emph{expected utility} of $X \in L^1$ is given by\footnote{The asymptotic property \eqref{eq:def:utility fn:not superlinear} implies that there exist $a,x_0 \in (0,\infty)$ such that
\begin{equation*}
u(x) \le ax, \quad \textnormal{for all } x \in (x_0,\infty).
\end{equation*}
As a result, for any $X\in L^1$, the expectation $\mathbb{E}[u(X)]$ is well-defined and takes values in $[-\infty,\infty)$.}
\begin{equation*}
    \cE_u(X):=\mathbb{E}[u(X)].
\end{equation*}
\end{example}

\begin{remark}
It is worth noting that our definition of a utility function and expected utility is very general. Up to normalisation, it includes any classical concave (von Neumann-Morgenstern) utility function, the concave-convex-concave utility functions discussed in \cite{FriedmanSavage1948,Markowitz1952b}, the star-shaped utility functions studied by \cite{LandsbergerMeilijson1990}, and the $S$-shaped utility functions studied by \cite{KahnemanTversky1979}. But it also includes many examples beyond these cases. In particular, since we understand increasing in a weak sense, we allow $u$ to have constant parts. Also note that we do not require $u$ to be continuous. Finally, expected utility will be concave or star-shaped when $u$ is concave or negatively star-shaped, respectively.
\end{remark}


\section{Sensitivity to large losses}
\label{sec: sensitivity to large losses}

In this section we introduce a new property of risk/utility functionals, called sensitivity to large losses, for which we provide a number of characterizations. As suggested by the terminology, a risk/utility functional is sensitive to large losses when it reacts to excessive loss levels. More precisely, sensitivity to large losses stipulates that it is not possible to rescale a position that exposes the agent to downside outcomes without losing acceptability once the position's size, hence the implied losses, becomes too large. For better comparison, we also introduce the property of sensitivity to losses, which prescribes that every position containing some loss potential is unacceptable.

\begin{definition}
Let $\cS\subset\cX$. A risk functional $\cR$, resp.\ a utility functional $\cU$, is called
{\em sensitive to losses} on $\cS$ if, for every $X\in\cS$ with $\probp(X<0)>0$, we have $\cR(X)>0$, resp.\ $\cU(X)<0$. It is called {\em sensitive to large losses} on $\cS$ if, for every $X\in\cS$ with $\probp(X<0)>0$, there exists $\lambda_X\in(0,\infty)$ such that $\cR(\lambda X)>0$, resp.\ $\cU(\lambda X)<0$, for each $\lambda\in(\lambda_X,\infty)$.
\end{definition}

\smallskip

It is clear that sensitivity to losses is a very stringent property, which implies sensitivity to large losses in many situations.

\begin{proposition}
Let $\cS\subset\cX$ and consider a risk functional $\cR$. Assume $\cR$ is star-shaped or $\cS$ is a cone, i.e., $\lambda X\in\cS$ for all $X\in\cS$ and $\lambda\in(0,\infty)$. If $\cR$ is sensitive to losses on $\cS$, then $\cR$ is sensitive to large losses on $\cS$.
\end{proposition}

\smallskip

The converse implication generally fails to hold as illustrated by the following example featuring expected exponential utility and the entropic risk measure.

\begin{example}
\label{ex: example zero}
Let $\cX= L^\infty$ and consider the exponential utility function $u(x):=1-e^{-x}$ for $x\in\R$. In addition, assume there is $A\in\cF$ with $\probp(A)=1/2$ and define $Y:=\ind_A-\frac{1}{3}\ind_{A^c}\in\cX$.

(a) The utility functional $\cE_u(X) := \E[u(X)]$ for $X \in \cX$ fails to be sensitive to losses on $\cX$ because $\cE_u(Y)=1-\frac{1}{2}(e^{-1}+e^{1/3})>0$. However, $\cE_u$ is sensitive to large losses on $\cX$. Indeed, take $X\in\cX$ with $\probp(X<0)>0$ and choose $\varepsilon\in(0,\infty)$ such that $\probp(X\le-\varepsilon)>0$. Letting $\lambda\to\infty$ gives
\[
\cE_u(\lambda X) \le \probp(X\le-\varepsilon)u(-\lambda\varepsilon)+\probp(X>0)u(\lambda\lnorm X\rnorm_\infty) \to -\infty.
\]

(b) Consider now the entropic risk measure defined for every $X\in\cX$ by
\[
\cR_u(X) := \inf\{m\in\R \,; \ \E[u(X+m)]\ge0\} = \log(\E[e^{-X}]).
\]
As $\cR_u(X) = \log(1-\cE_u(X))$, (a) implies that $\cR_u$ is sensitive to large losses but not to losses on $\cX$.
\end{example}

\smallskip

%
We now start our study of equivalent formulations of sensitivity to large losses. The case of a positively homogeneous risk/utility functional can be settled immediately. Under positive homogeneity, sensitivity to large losses is tantamount to sensitivity to losses and therefore stipulates that any financial position that exposes the agent to some loss cannot be acceptable regardless of how it is rescaled.

\begin{proposition}
\label{prop: positively homogeneous}
Let $\cS\subset\cX$. For a positively homogeneous risk functional $\cR$ the following statements are equivalent:
\begin{enumerate}
    \item[(a)] $\cR$ is sensitive to large losses on $\cS$.
    \item[(b)] $\cR$ is sensitive to losses on $\cS$.
\end{enumerate}
\end{proposition}

\smallskip

The previous result has the immediate and important consequence that neither of the two most prominent risk functionals used in practice, Value at Risk and Expected Shortfall, is sensitive to large losses on their domain of definition. While this should be expected from Value at Risk due to its well-known blindness to tail risk, the statement might sound counterintuitive for Expected Shortfall, which is often viewed as a risk functional that is sensitive to the loss tail. This shows that ``loss sensitivity'' is a multi-faceted concept calling for an exploration at different levels in order to obtain a more comprehensive understanding of the merits and drawbacks of a given risk functional. In Sections~\ref{subsect: other loss sensitivity} and \ref{sect: examples} we will see how and to what extent Value at Risk and Expected Shortfall become sensitive to large losses if they are either localized to suitable domains of financial positions or properly adjusted.

\begin{example}
\label{ex: var and es}
Let $\cX=L^\infty$ and take $\alpha\in(0,1)$. (a) It is immediate to see that the tail blindness of $\VaR_\alpha$ makes it insensitive to large losses already on $\cX$. Indeed, assume there is $A\in\cF$ with $\probp(A)\in(0,\alpha]$. Let $X:=-\ind_A\in\cX$ and observe that $\probp(X<0)=\probp(A)>0$ but $\VaR_\alpha(X)=0$, showing that $\VaR_\alpha$ fails to be sensitive to losses, hence to large losses (by positive homogeneity).

\smallskip

(b) In spite of being sensitive to the tail by definition, $\ES_\alpha$ is not sensitive to large losses on $\cX$. Indeed, assume there is $A\in\cF$ with $\probp(A)\in(0,\alpha)$. Let $n\in\N$ with $n\ge\frac{\probp(A)}{\alpha-\probp(A)}$ and define $X:=-\ind_A+n\ind_{A^c}\in\cX$. Clearly, $\probp(X<0)=\probp(A)>0$. However,
\[
\ES_\alpha(X) = \tfrac{1}{\alpha}[\probp(A)-n(\alpha-\probp(A))] \le 0.
\]
Thus, $\ES_\alpha$ fails to be sensitive to losses, hence to large losses (by positive homogeneity).
\end{example}

\smallskip

For the analysis of sensitivity to large losses of a risk/utility functional that is not positively homogeneous, it is convenient to work with its recession functional, i.e., the smallest/largest positively homogeneous risk/utility functional that lies above/below the given functional in a pointwise way. 
\begin{definition}
For a risk functional $\cR$, define the risk functional $\cR^\infty$ by
\[
\cR^\infty(X) := \sup_{\lambda\in(0,\infty)}\frac{\cR(\lambda X)}{\lambda}.
\]
For a utility functional $\cU$, define the utility functional $\cU^\infty$ by
\[
\cU^\infty(X) := \inf_{\lambda\in(0,\infty)}\frac{\cU(\lambda X)}{\lambda}.
\]
\end{definition}
Note that a positively homogeneous risk/utility functional coincides with its recession functional. Note also that there is no ambiguity about the definition of recession functional as the zero functional is the only functional that is both a risk and utility functional. The next proposition shows that the recession functional is sensitive to large losses, or equivalently sensitive to losses, whenever the underlying risk/utility functional is sensitive to large losses, while the converse is not true in general as illustrated by the example below.

\begin{proposition}
\label{prop: first link with recession map}
Let $\cS\subset\cX$ and $\cR$ be a risk functional. The following statements are equivalent:
\begin{enumerate}
    \item[(a)] $\cR^\infty$ is sensitive to large losses on $\cS$.
    \item[(b)] $\cR^\infty$ is sensitive to losses on $\cS$.
    \item[(c)] For every $X\in\cS$ with $\probp(X<0)>0$ there is $\lambda\in(0,\infty)$ such that $\cR(\lambda X)>0$.
\end{enumerate}
In particular, $\cR^\infty$ is sensitive to losses on $\cS$ whenever $\cR$ is sensitive to large losses on $\cS$.
\end{proposition}

\smallskip

\begin{example}
\label{ex: recession functional}
Let $\cX = L^\infty$. (a) Consider the risk functional given for every $X\in\cX$ by
\[
\cR(X) := \min\{\esssup(-X),\VaR_\alpha(X)+1\}.
\]
It is easy to verify that $\cR^\infty$ is sensitive to losses on $\cX$ because, for every $X\in\cX$,
\[
\cR^\infty(X) = \sup_{\lambda\in(0,\infty)}\min\{\esssup(-X),\VaR_\alpha(X)+\tfrac{1}{\lambda}\} = \esssup(-X).
\]
However, $\cR$ is not sensitive to large losses on $\cX$. Indeed, consider the position $X:=-\ind_A\in\cX$ with $\probp(A)\in(0,\alpha)$ and note that $\probp(X<0)>0$ and $\cR(\lambda X)=\min\{0,-\lambda+1\}$ for every $\lambda\in(0,\infty)$. In particular, $\cR(\lambda X)\le0$ for $\lambda\in[1,\infty)$. Incidentally, note that $\cR$ is not star-shaped because, e.g., $\cR(2X)=-1<0=2\cR(X)$.

\smallskip

(b) Consider the utility functional given for every $X\in\cX$ by
\[
\cU(X) := \max\{\E[X\ind_{\{X<0\}}],\E[X]-1\}.
\]
It is easy to verify that $\cU^\infty$ is sensitive to losses on $\cX$ because, for every $X\in\cX$,
\[
\cU^\infty(X) = \inf_{\lambda\in(0,\infty)}\max\{\E[X\ind_{\{X<0\}}],\E[X]-\tfrac{1}{\lambda}\} = \E[X\ind_{\{X<0\}}].
\]
However, $\cU$ is not sensitive to large losses on $\cS$. Indeed, consider the position $X:=\ind_{A^c}-\ind_A\in\cS$ with $\probp(A)\in(0,1/4)$ and note that $\probp(X<0)>0$ and $\cU(\lambda X)=\max\{-\lambda\probp(A),\lambda(1-2\probp(A))-1\}$ for every $\lambda\in(0,\infty)$. In particular, $\cU(\lambda X)\ge0$ for $\lambda\in[1/(1-2\probp(A)),\infty)$. Incidentally, note that $\cU$ is not star-shaped because, e.g., $\cU(2X)=1-4\probp(A)>-2\probp(A)=2\cU(X)$.
\end{example}


\subsection{Star-shaped risk/utility functionals}

While sensitivity to large losses of a risk/utility functional is in general not equivalent to that of its associated recession functional, this becomes true in the broad class of star-shaped risk/utility functionals. This is convenient because it gives a simple test to verify sensitivity to large losses in concrete examples; see Section~\ref{sect: examples}. In the following result we establish this and a variety of other necessary and sufficient conditions for sensitivity to large losses that are economically meaningful. In particular, for a star-shaped risk/utility functional, sensitivity to large losses holds whenever the risk/utility of a position exposing the agent to some loss increases/decreases without limit with the position's size.

\begin{theorem}
\label{theo: recession functional}
Let $\cS\subset\cX$. For a star-shaped risk functional $\cR$, the following  are equivalent:
\begin{enumerate}
    \item[(a)] $\cR$ is sensitive to large losses on $\cS$.
    \item[(b)] For every $X\in\cS$ with $\probp(X<0)>0$, we have $\sup_{\lambda> 0}\cR(\lambda X) =\infty$.
    \item[(c)] For every $X\in\cS$ with $\probp(X<0)>0$, there exists $\lambda\in(0,\infty)$ such that $\cR(\lambda X)>0$.
    \item[(d)] $\cR^\infty$ is sensitive to losses on $\cS$.
\end{enumerate}
Moreover, if $\cR$ is sensitive to large losses on $\cS$, then the following statement holds:
\begin{enumerate}
    \item[(e)] For every $X\in\cS$ with $\probp(X<0)>0$, there exists $Y\in\cX$ such that $\cR(X+Y)>\cR(Y)$.
\end{enumerate}
\end{theorem}

\smallskip

As shown in Theorem~\ref{theo: recession functional}, under a star-shaped risk/utility functional that is sensitive to large losses, it is not possible to reduce/increase the risk/utility of an arbitrary financial position by adding to it a position that exposes the agent to some losses. The next example shows that, for a star-shaped risk/utility functional, this condition is necessary but, in general, not sufficient to ensure sensitivity to large losses. We illustrate this by considering expected utility in the presence of a star-shaped utility function. In fact, this is also true if star-shapedness is replaced with the stronger positive homogeneity, as shown by Value at Risk.

\begin{example}
Let $\cX=L^\infty$. (a) Take $\alpha\in(0,1)$. Recall that $\VaR_\alpha$ is positively homogeneous and, hence, star-shaped. Moreover, $\VaR_\alpha$ is not sensitive to large losses on $\cX$. However, for every $X\in\cX$ with $\probp(X<0)>0$ there exists $Y\in\cX$ with $\VaR_\alpha(Y)=0$ such that $\VaR_\alpha(X+Y)>0$. Indeed, if $\probp(X<0)>\alpha$, then it suffices to take $Y=0$; otherwise, we can take $Y=-(1+\lnorm X\rnorm_\infty)\ind_A$ where $A\in\cF$ satisfies $A\subset\{X\ge 0\}$ and $\alpha-\probp(X<0)<\probp(A)\le\alpha$.

\smallskip

(b) Consider the star-shaped utility function given by
\[
u(x):=
\begin{cases}
1-e^{-x} & \mbox{if} \ x\ge0,\\
0 & \mbox{if} \ x<0.
\end{cases}
\]
Then $\cE_u$ is star-shaped but clearly not sensitive to large losses. However, for every $X\in\cX$ with $\probp(X<0)>0$ there exists $Y \in \cX$ such that $\cE_u(X + Y) < \cE_u(Y)$. Indeed, choose $\varepsilon\in(0,\infty)$ such that $\probp(X\le-\varepsilon)>0$. 
For each $n\in\N$ define $Y_n:= \Vert X \Vert_\infty +(n - X)\ind_{\{X>0\}}$ and note that 
\[
\cE_u(X+Y_n) \le \probp(X\le-\varepsilon)u(\Vert X \Vert_\infty-\varepsilon)+\probp(-\varepsilon<X\le0)u(\Vert X \Vert_\infty)+\probp(X>0)u(\Vert X \Vert_\infty+n),
\]
\[
\cE_u(Y_n) \ge \probp(X\le 0)u(\Vert X \Vert_\infty)+\probp(X>0)u(n).
\]
As a result, taking $n$ large enough yields
\begin{align*}
\cE_u(Y_n)-\cE_u(X+Y_n) &\ge \probp(X\le-\varepsilon)[u(\Vert X \Vert_\infty)-u(\Vert X \Vert_\infty-\varepsilon)]+\probp(X>0)[u(n)-u(\Vert X \Vert_\infty+n)] \\
&= \probp(X\le-\varepsilon)e^{-\Vert X \Vert_\infty}(e^\varepsilon-1)-\probp(X>0)e^{-n}(1-e^{-\Vert X \Vert_\infty}) > 0.
\end{align*}
\end{example}

\smallskip

If star-shapenedness is replaced with the stronger property of convexity/concavity, and we assume the risk/utility functional is lower/upper semicontinuous, then we obtain a further characterization of sensitivity to large losses in terms of \emph{risk reduction/utility increase}. More precisely, sensitivity to large losses holds if, and only if, every position that exposes the agent to some loss will increase/reduce the risk/utility of some other position after aggregation. In other words, sensitivity to large losses fails precisely when there exists a position with some loss potential that can be merged into any other position without increasing/reducing risk/utility.

\begin{theorem}
\label{theo: convex}
Let $\cS\subset\cX$. For a risk functional $\cR$ that is convex and lower semicontinuous the following statements are equivalent:
\begin{enumerate}
    \item[(a)] $\cR$ is sensitive to large losses on $\cS$.
    \item[(b)] For every $X\in\cS$ with $\probp(X<0)>0$ there exists $Y\in\cX$ such that $\cR(X+Y)>\cR(Y)$.
\end{enumerate}
\end{theorem}


\subsection{Cash-additive risk/utility functionals}
\label{subsec:cash additive star-shaped risk measures}

In this section we focus on sensitivity to large losses in the broad class of cash-additive risk/utility functionals. Our first result shows that, when cash-additivity is coupled with star-shapedness and the chosen domain is stable under cash injections, sensitivity to large losses is equivalent to the recession functional reducing to an essential supremum/infimum. This gives a simple way to test for sensitivity to large losses for concrete examples of risk/utility functionals; see Section~\ref{sect: examples}. In addition, it reinforces that Value at Risk and Expected Shortfall cannot be sensitive to large losses on a domain that is ``too large''; see Example~\ref{ex: var and es}.

\begin{theorem}
\label{theo: cash-additive}
Let $\cS\subset\cX$ satisfy $\cS+\R\subset\cS$. For a cash-additive and star-shaped risk functional $\cR$ the following statements are equivalent:
\begin{enumerate}
    \item[(a)] $\cR$ is sensitive to large losses on $\cS$.
    \item[(b)] $\cR^\infty(X)=\esssup(-X)$ for every $X\in\cS$.
\end{enumerate}
\end{theorem}

\smallskip

As a direct corollary of Theorem \ref{theo: cash-additive} it follows that the essential supremum/infimum is the unique cash-additive risk/utility functional that is simultaneously positively homogeneous and sensitive to large losses on a domain that is stable under cash injections.

\begin{corollary}
\label{cor: cash-additive}
Let $\cS\subset\cX$ satisfy $\cS+\R\subset\cS$. For a cash-additive and positively homogeneous risk functional $\cR$ the following statements are equivalent:
\begin{enumerate}
    \item[(a)] $\cR$ is sensitive to large losses on $\cS$.
    \item[(b)] $\cR(X)=\esssup(-X)$ for every $X\in\cS$.
\end{enumerate}
\end{corollary}

\smallskip

\begin{remark}
It should be noted that neither of the above results hold in the case that the chosen domain of random variables fails to be stable under cash injections. For instance, for every $\alpha\in(0,1)$ we immediately see that $\ES_\alpha$ is sensitive to large losses on $\cS=\{X\in L^1 \,; \ X<0\}$; see also the results in Section~\ref{subsect: other loss sensitivity}.
\end{remark}

\smallskip

If we maintain cash additivity but replace positive homogeneity with convexity/concavity, and lower/upper semicontinuity is satisfied, then we obtain the following characterization of sensitivity to large losses on general domains of random variables as a direct consequence of Theorem~\ref{theo: convex}. More explicitly, the result shows that sensitivity to large losses fails whenever there exists an acceptable position with some loss potential that remains acceptable after aggregation with any position with zero risk/utility. The explanation here is as follows. A position with zero risk/utility lies right at the edge of acceptability, so it likely contains profits as well as losses. It is conceivable that merging such a position with another position that already contains losses could result into losing acceptability, especially when (part of) the respective profits are drained by the losses of the other position, so that merging turns out creating/draining extra risk/utility. This happens precisely when the reference risk/utility functional is sensitive to large losses.

\begin{corollary}
Let $\cS\subset\cX$. For a cash-additive, convex, and lower semicontinuous risk functional $\cR$ the following statements are equivalent:
\begin{enumerate}
    \item[(a)] $\cR$ is sensitive to large losses on $\cS$.
    \item[(b)] For every $X\in\cS$ with $\probp(X<0)>0$ there is $Y\in\cX$ such that $\cR(Y)=0$ and $\cR(X+Y)>0$.
\end{enumerate}
\end{corollary}

\smallskip

We conclude this section by showing that sensitivity to large losses for a generic (not necessarily cash-additive) risk/utility functional can often be characterized by looking at a special cash-additive risk/utility functional, thus highlighting the special role played by cash-additive risk/utility functionals in the study of sensitivity to large losses. This will prove very useful in Sections~\ref{sect: examples} and~\ref{sec:examples utility} where we study sensitivity to large losses for a number of explicit risk/utility functionals. We start by introducing some standard notation that will be used to state the announced result.

\begin{definition}
For $\cA\subset\cX$ we define the functional $\cR_\cA:\cX\to[-\infty,\infty]$ by
\[
\cR_\cA(X) := \inf\{m\in\R \,; \ X+m\in\cA\}.
\]
Similarly, we define the functional $\cU_\cA:\cX\to[-\infty,\infty]$ by
\[
\cU_\cA(X) := \sup\{m\in\R \,; \ X-m\in\cA\}.
\]
\end{definition}

It is clear that every cash-additive risk functional $\cR$ coincides with $\cR_{\cA_\cR}$ and that every cash-additive utility functional $\cU$ coincides with $\cU_{\cA_\cU}$. It is also clear that $\cR_\cA$ and $\cU_\cA$ do not generally qualify as a risk and utility functional according to our definition. We refer to Proposition~\ref{prop: rho A risk measure} for an equivalent condition for this to hold in the case that interests us, namely when $\cA$ is the acceptance set of a given risk/utility functional (which need not be cash-additive). As a first step, we provide a sufficient condition based on cash-additivity for a generic (not necessarily cash-additive) risk/utility functional to be sensitive to large losses.

\begin{proposition}
\label{prop: sensitivity via acceptability preliminary}
Let $\cS\subset\cX$ and consider a risk functional $\cR$ such that $\cR_{\cA_\cR}$ is also a risk functional. If $\cR_{\cA_\cR}$ is sensitive to large losses on $\cS$, then $\cR$ is also sensitive to large losses on $\cS$.
\end{proposition}

\smallskip

The next example shows that the converse implication fails in general, even if the reference risk/utility functional is assumed to be positively homogeneous and convex/concave.

\begin{example}
Let $\cX=L^1$. Take $A\in\cF$ with $\probp(A)\in(0,1)$ and set $\cS := \{X \in L^1 \,; \ \E[X\ind_A]\le0\}$. Moreover, for every $X\in L^1$, define
\begin{equation*}
\cR(X) := \begin{cases}
\infty, & \mbox{if} \ X\in\cS,\\
0, & \mbox{otherwise}.
\end{cases}
\end{equation*}
Clearly, $\cR$ is positively homogeneous, convex, and sensitive to (large) losses on $\cS$. It is also clear that $\cR_{\cA_\cR}$ is a risk functional. Indeed, for every $X\in L^1$,
\[
\cR_{\cA_\cR}(X) = \inf\{m\in\R \,; \ \E[(X+m)\ind_A]>0\} = \E[-X\vert A].
\]
However, $\cR_{\cA_\cR}$ fails to be sensitive to large losses on $\cS$ because, e.g., the position $X:=-\ind_{A^c}\in L^1$ belongs to $\cS$ but $\cR_{\cA_\cR}(\lambda X)=0$ for every $\lambda\in(0,\infty)$.
\end{example}

\smallskip

To conclude this section we explore under which conditions the converse implication in Proposition~\ref{prop: sensitivity via acceptability preliminary} is true. We start by showing that, under the mild property of lower/upper regularity, the converse implication always holds and, hence, whether or not a risk/utility functional is sensitive to large losses depends on the sensitivity to large losses of the cash-additive risk/utility functional associated with the acceptance set of the original risk/utility functional. Recall that any risk/utility functional that is either lower/upper semicontinuous or cash-additive is automatically lower/upper regular according to our definition.

\begin{theorem}
\label{the: sensitivity via acceptability v2}
Let $\cS\subset\cX$. For every lower regular risk functional $\cR$ such that $\cR_{\cA_\cR}$ is also a risk functional the following statements are equivalent:
\begin{enumerate}
    \item[(a)] $\cR$ is sensitive to large losses on $\cS$.
    \item[(b)] $\cR_{\cA_\cR}$ is sensitive to large losses on $\cS$.
\end{enumerate}
\end{theorem}

\smallskip

As a next result, we establish the converse implication for star-shaped risk/utility functionals when sensitivity to large losses is postulated over a domain that is stable under cash injections. 

\begin{theorem}
\label{thm:sensitivity via acceptability v3}
Let $\cS\subset\cX$ satisfy $\cS+\R=\cS$. For every star-shaped risk functional $\cR$ such that $\cR_{\cA_\cR}$ is also a risk functional the following statements are equivalent:
\begin{enumerate}
    \item[(a)] $\cR$ is sensitive to large losses on $\cS$.
    \item[(b)] $\cR_{\cA_\cR}$ is sensitive to large losses on $\cS$.
\end{enumerate}
\end{theorem}

\subsection{Localized versions}
\label{subsect: other loss sensitivity}

We have defined the notion of sensitivity to (large) losses with respect to a subset $\cS$ of the whole space $\cX$. The rationale for doing so is to give flexibility to that notion by capturing other special versions of ``sensitivity to large losses''. In this section we illustrate this by focusing on three canonical choices for $\cS$, two of which have been considered in the extant literature.

\begin{definition}
A risk functional $\cR$, resp.\ a utility functional $\cU$, is said to be
\begin{enumerate}
    \item[(i)] {\em sensitive to sure losses} if, for every $X\in\cX$ with $X<0$, we have $\cR(X) > 0$, resp.\ $\cU(X) < 0$. It is called {\em sensitive to large sure losses} if, for every $X\in\cX$ with $X<0$, there exists $\lambda_X\in(0,\infty)$ such that $\cR(\lambda X) > 0$, resp.\ $\cU(\lambda X) < 0$, for each $\lambda\in(\lambda_X,\infty)$. 
    \item[(ii)] {\em sensitive to pure losses} if, for every nonzero $X\in\cX$ with $X\leq 0$, we have $\cR(X) > 0$, resp.\ $\cU(X) < 0$. It is called {\em sensitive to large pure losses} if, for every nonzero $X\in\cX$ with $X\le0$, there exists $\lambda_X\in(0,\infty)$ such that $\cR(\lambda X) > 0$, resp.\ $\cU(\lambda X) < 0$, for each $\lambda\in(\lambda_X,\infty)$.
    \item[(iii)] {\em sensitive to expected losses} if, for every nonzero $X\in\cX$ with $\E[X] \leq 0$, we have $\cR(X) > 0$, resp.\ $\cU(X) < 0$. It is called {\em sensitive to large expected losses} if, for every nonzero $X\in\cX$ with $\E[X]\le0$, there exists $\lambda_X\in(0,\infty)$ such that $\cR(\lambda X) > 0$, resp.\ $\cU(\lambda X) < 0$,  for each $\lambda\in(\lambda_X,\infty)$.
\end{enumerate}
\end{definition}

\begin{remark}
While sensitivity to (large) sure losses does not appear to have been explicitly investigated in the literature before, sensitivity to pure losses was studied under the name of \emph{relevance} in the context of risk functionals in \cite{artzner1999coherent, delbaen2002coherent}. The property of sensitivity to large expected losses was recently considered in \cite{HerdegenKhan2024}  to prove the existence of optimal portfolios in a mean-risk framework extending the classical mean-variance Markowitz's setting.
\end{remark}

\smallskip

It is readily seen that sensitivity to (large) sure/pure/expected losses corresponds to sensitivity to (large) losses on the three domains
\[
\cX_{--}:=\{X\in\cX \,; \ X<0\}, \ \ \ \ \cX_{-}:=\{X\in\cX \,; \ X\le0\}, \ \ \ \ \cX_e:=\{X\in\cX \,; \ \E[X]\le0\},
\]
respectively. Since these three domains are ordered by inclusion, we have the following result. In the example below we show that none of the converse implications generally hold.

\begin{proposition}
Let $\cR$ be a risk functional.
\begin{enumerate}
\item[(a)] If $\cR$ is sensitive to (large) expected losses, it is sensitive to (large) pure losses.
\item[(b)] If $\cR$ is sensitive to (large) pure losses, it is sensitive to (large) sure losses.
\end{enumerate}
\end{proposition}

\smallskip

\begin{example}
\label{ex: var}
(a) Let $\cX=L^\infty$ and $\alpha\in(0,1)$. Value at Risk generally fails to be sensitive to (large) pure losses and hence also to (large) expected losses. Indeed, assume there is $A \in \cF$ with $\probp(A)\in(0,\alpha)$ and set $X := -\ind_A\in\cX$ to derive $\VaR_\alpha(X) = 0$. Nevertheless, it is sensitive to (large) sure losses. Indeed, for every $X\in\cX$ with $X < 0$ there exists $m\in(0,\infty)$ such that $\P(X +m <0) > \alpha$, whence $\VaR_\alpha(X) > 0$.

\smallskip

(b) Consider the $S$-shaped utility function given by
\begin{equation*}
u(x) := 
\begin{cases}
\sqrt{x} & \text{if } x \geq 0, \\
-\sqrt{-x} & \text{if } x < 0.
\end{cases}
\end{equation*}
Then the expected utility functional $\cE_u$ is clearly sensitive to (large) pure losses but fails to be sensitive to large expected losses because for every $X \in \cX$ whose distribution is symmetric at $0$, we have $\E[X] =0$ but $\cE_u(\lambda X) = 0$ for every $\lambda\in(0,\infty)$.
\end{example}

\smallskip

The following result gives practical sufficient conditions for sensitivity to large pure and expected losses for star-shaped risk/utility functionals in terms of their recession functionals being suitably bounded by the expected loss of a position.

\begin{proposition}
\label{prop:sens large exp/pure:1}
Let $\cR$ be a star-shaped risk functional.
\begin{enumerate}
\item[(a)] If $\cR^\infty(X) > \E[-X]$ for every nonconstant $X\in\cX$, then $\cR$ is sensitive to large expected losses.
\item[(b)] If $\cR^\infty(X) \geq \E[-X]$ for every $X\in\cX$, then $\cR$ is sensitive to large pure losses.
\end{enumerate}
\end{proposition}

\smallskip

For a star-shaped risk/utility functional that is also cash-additive the previous sufficient condition for sensitivity to large expected losses is also necessary. By contrast, this is not true for sensitivity to large pure losses even if star-shapedness is replaced by positive homogeneity and convexity, as shown by the example below.

\begin{theorem}
\label{thm:sens large exp/pure:2}
Let $\cR$ be a star-shaped and cash-additive risk functional. Then
$\cR$ is sensitive to large expected losses if and only if $\cR^\infty(X) > \E[-X]$ for all non-constant $X \in \cX$.
\end{theorem}

\smallskip

\begin{example}
Let $\cX=L^\infty$. Take $A\in\cF$ such that $\probp(A)\in(1/2,1)$ and consider the risk functional given by $\cR(X):=(2-1/\probp(A))\E[-X\ind_A]+2\E[-X\ind_{A^c}]$ for $X\in\cX$. It is straightforward to check that $\cR$ is positively homogenous, convex and cash-additive. In addition, $\cR$ is clearly sensitive to (large) pure losses. However, $\cR$ does not satisfy the condition in Proposition \ref{prop:sens large exp/pure:1} (b) because for $Y:=-\ind_A+\frac{\probp(A)}{\probp(A^c)}\ind_{A^c}$ we get $\cR^\infty(Y)=\cR(Y)=-1<0=\E[-Y]$.
\end{example}

\smallskip

The following corollary shows that for positively homogeneous and cash-additive risk measures, sensitivity to (large) expected losses is equivalent to strict expectation boundedness.

\begin{corollary}
\label{cor:sens large exp/pure:2}
Let $\cR$ be a positively homogeneous and cash-additive risk functional. Then $\cR$ is sensitive to (large) expected losses if and only if $\cR(X) > \E[-X]$ for all non-constant $X \in \cX$.
\end{corollary}

\smallskip

\begin{example}
\label{ex: ES}
Unlike Value at Risk, Expected Shortfall at any level $\alpha \in (0, 1)$ is sensitive to (large) expected and hence (large) pure/sure losses. This follows directly from Corollary \ref{cor:sens large exp/pure:2} and the fact $\ES_\alpha$ is strictly expectation bounded. Indeed, a slight simplification of the proof of \cite[Lemma 3.6]{KochMunariSvindland2018} shows that, for every non-constant $X\in L^1$,
\begin{align*}
\E[-X] &= \int_0^1\VaR_\beta(X)d\beta = \alpha\ES_\alpha(X)+\int_\alpha^1\VaR_\beta(X)d\beta \\
&\le \alpha\ES_\alpha(X)+(1-\alpha)\VaR_\alpha(X) \le \ES_\alpha(X),
\end{align*}
where at least one of the two inequalities must be strict for otherwise $X$ would be constant.
\end{example}

\subsection{Sensitivity to loss concentrations}
\label{sect: concentrations}
We end this section by looking at a property that is closely related to sensitivity to large losses in that it axiomatizes sensitivity to the accumulation or concentration of losses in specific scenarios. To the best of our knowledge, this property, which is named sensitivity to loss concentrations, has never been explicitly investigated in the literature.

\begin{definition}
A risk functional $\cR:\cX\to(-\infty,\infty]$, resp.\ a utility functional $\cU:\cX\to[-\infty,\infty)$, is called {\em sensitive to loss concentrations} if, for all $X\in\cX$ and $A\in\cF$ with $\probp(A)>0$, there exists $\lambda_X \in(0,\infty)$ such that $\cR(X-\lambda\ind_A)>0$, resp.\ $\cU(X-\lambda\ind_A)<0$, for each $\lambda \in (\lambda_X, \infty)$.
\end{definition}

\smallskip

The relationship between sensitivity to loss concentrations and sensitivity to large losses is delicate. The aim of this section is to show that sensitivity to loss concentrations strictly lies in between sensitivity to large expected losses and sensitivity to large pure losses. We begin our discussion by showing that sensitivity to loss concentrations is implied by sensitivity to large expected losses. As illustrated below, the converse implication generally fails. This is true even if one requires positive homogeneity and convexity/concavity.

\begin{proposition}
\label{prop: loss concentrations}
Any risk functional $\cR$ that is sensitive to large expected losses is sensitive to loss concentrations.
\end{proposition}

\smallskip

\begin{example}
Let $\cX=L^\infty$ and consider the risk functional defined by $\cR(X):=\E[-X]$ for $X\in\cX$. Then $\cR$ is clearly sensitive to loss concentrations. However, it is not sensitive to large expected losses because any nonzero $X \in \cX$ with $\E[X] =0$ satisfies $\cR(\lambda X) = 0$ for each $\lambda > 0$.
\end{example}

We proceed to show that sensitivity to loss concentrations implies sensitivity to large pure losses. Again the converse generally fails as illustrated by the example below.

\begin{proposition}
\label{prop: necessary for sensitivity to loss concentrations}
Any risk functional $\mathcal{R}$ that is sensitive to loss concentrations is sensitive to large pure losses.
\end{proposition}

\smallskip

\begin{example}
Let $\cX=L^\infty$. Take $A\in\cF$ such that $\probp(A)\in(0,1)$ and define a risk functional by $\cR(X):=\min\{\esssup(-X),\esssup((1-X)\ind_A)\}$ for $X\in\cX$. Note that $\cR$ is sensitive to (large) pure losses because, for every nonzero $X\in\cX$ with $X\le 0$, we have $\cR(X)\ge\min\{\esssup(-X),1\}>0$. However, setting $X:=\ind_A$ we see that $\cR(X-\lambda\ind_{A^c})=\min\{\lambda,0\}=0$ for every $\lambda\in(0,1)$, showing that $\cR$ is not sensitive to loss concentrations.
\end{example}

\smallskip

Note that in the counterexample above, the risk functional is not star-shaped. We conclude this section by showing that for a star-shaped risk/utility functional that is lower/upper semicontinuous, sensitivity to loss concentrations is equivalent to sensitivity to large pure losses.

\begin{theorem}
\label{thm: loss concentrations part two}
A star-shaped and lower semicontinuous risk functional $\cR$ is sensitive to loss concentrations if and only if it is sensitive to large pure losses.
\end{theorem}


\section{Examples I: Risk functionals}
\label{sect: examples}

We have already seen in Example \ref{ex: var and es} that the two most prominent risk functionals in practice, Value at Risk and Expected Shortfall, fail to be sensitive to large losses on $L^\infty$ (and a fortiori on~$L^1$). Nevertheless, we have also seen that they are sensitive to large losses on suitable subdomains of $L^\infty$, cf.~Examples \ref{ex: var} and \ref{ex: ES}. In this section we focus on other risk functionals encountered in the literature, some of which are modifications of Value at Risk and Expected Shortfall, and use our results to determine whether they are sensitive to large losses or not. Throughout the entire section we assume that, for each $\varepsilon \in (0, 1)$, there exists $A \in \cF$ with $0 < \mathbb{P}(A) \leq  \varepsilon$. This property is implied by, but far weaker than, the requirement of $(\Omega,\cF,\probp)$ being non-atomic. To keep the analysis self-contained, we investigate sensitivity to large losses only on the domains $L^1$ and $L^\infty$, which are equipped with their canonical topological lattice structure.

\subsection{Loss Value at Risk}

Let $\alpha:(-\infty,0]\to[0,1)$ be increasing. The {\em Loss Value at Risk} of $X \in L^1$ with \emph{benchmark loss distribution} $\alpha$ is defined by (see \cite{BignozziBurzoniMunari2020})
\[
\LVaR_\alpha(X) := \sup_{\ell\in(-\infty,0]}\{\VaR_{\alpha(\ell)}(X)+\ell\}, \quad \textnormal{where} \quad 
\VaR_0(X):=\esssup(-X).
\]
This is a cash-additive and star-shaped risk functional according to our general definition. The function $\alpha$ can be seen as a threshold distribution for acceptability because a position $X \in L^1$ is acceptable if and only if $\probp(X < \ell) \leq \alpha(\ell)$ for all $\ell\in (-\infty, 0]$. In particular, if there exists $\alpha_0\in(0,1)$ such that $\alpha(\ell)=\alpha_0$ for every $\ell \in(-\infty,0]$, then $\LVaR_\alpha$ coincides with Value at Risk at level $\alpha_0$. This implies that $\LVaR_\alpha$ cannot always be sensitive to large losses. However, the next proposition shows that sensitivity to large losses holds whenever the function $\alpha$ is not bounded away from zero. If we recall that Value at Risk is not sensitive to large losses on $L^\infty$, then we can interpret this result as showing how to modify Value at Risk by assigning specific weights to each loss level in order to gain sensitivity to large losses on the whole space.

\begin{proposition}
\label{prop: loss var}
Let $\alpha:(-\infty,0]\to(0,1)$ be increasing and set $\underline{\alpha}:=\inf_{\ell \in(-\infty,0]}\alpha(\ell)$. The following are equivalent:
\begin{enumerate}
    \item[(a)] $\LVaR_\alpha$ is sensitive to large losses on $L^1$.
    \item[(b)] $\LVaR_\alpha$ is sensitive to large losses on $L^\infty$.
    \item[(c)] $\underline{\alpha}=0$.
\end{enumerate}
\end{proposition}

\subsection{Adjusted Expected Shortfall}

Let $g:(0,1]\to[0,\infty]$ be decreasing with $g(1) = 0$. The {\em adjusted Expected Shortfall} of $X \in L^1$ with \emph{risk profile} $g$ is defined by (see \cite{BurzoniMunariWang2022}\footnote{Note that because we work on $L^1$, we have changed the domain of $g$ so that it does not include $0$. Nevertheless all results in \cite{BurzoniMunariWang2022} carry over.})
\[
\ES^g(X) := \sup_{\alpha\in(0,1]}\{\ES_\alpha(X)-g(\alpha)\}, \quad \textnormal{where} \quad \ES_1(X):=\mathbb{E}[-X].
\]
This is a cash-invariant and convex risk functional according to our general definition. The function $g$ can be seen as the threshold between acceptable and unacceptable ES profiles because a position $X \in L^1$ is acceptable if and only if $\ES_\alpha(X) \leq g(\alpha)$ for all $\alpha\in (0, 1]$. For given $\alpha\in(0,1]$, if $g=\infty$ on $(0,\alpha)$ and $g=0$ on $[\alpha,1]$, then $\ES^g$ coincides with Expected Shortfall at level $\alpha$. As such, $\ES^g$ cannot generally be sensitive to large losses on $L^1$. However, we can always ensure sensitivity to large losses by choosing a finite function $g$. In other words, this gives a way to modify Expected Shortfall in order to make it sensitive to large losses on the whole space.

\begin{proposition}
\label{prop: adjusted es}
Let $g:(0,1]\to[0,\infty]$ be decreasing with $g(1)=0$. The following are equivalent:
\begin{enumerate}
    \item[(a)] $\ES^g$ is sensitive to large losses on $L^1$.
    \item[(b)] $\ES^g$ is sensitive to large losses on $L^\infty$.
    \item[(c)] $g$ is finite everywhere.
\end{enumerate}
\end{proposition}

\subsection{Shortfall risk measures}

For an increasing function $\ell:\R \to (-\infty, \infty]$ satisfying $\ell(0) = 0$ and $\limsup_{x \to -\infty} \tfrac{\ell(x)}{x} < \infty$, the {\em shortfall risk measure} of $X\in L^1$ is defined by (see \cite{FoellmerSchied2016})
\begin{equation*}
    \cR_\ell(X) := \inf \{ m\in\R \,; \ X+m \in \cA_\ell \}, \quad \cA_\ell:=\{X\in L^1 \,; \ \E[\ell(-X)] \leq 0\}.
\end{equation*}
The function $\ell$ is referred to as a \emph{loss function}. It reflects how risk averse an agent is. If $\ell(x)>0$ for every $x\in(0,\infty)$, then $\cR_\ell$ is a cash-additive risk functional in the sense of our general definition. 

\smallskip

The next result shows that whether a shortfall risk measure is sensitive to large losses or not depends only on the ratio of the tails of the underlying loss function. More precisely, sensitivity to large losses is fulfilled whenever the gain tail asymptotically dominates the loss tail.

\begin{proposition}
\label{prop:SR SLL}
Let $\ell$ be a loss function that is positively star-shaped (e.g. convex) on $\mathbb{R}_-$ and assume $\ell(x) > 0$ for every $x \in (0, \infty)$. The following are equivalent:
\begin{enumerate}
    \item[(a)] $\cR_\ell$ is sensitive to large losses on $L^1$.
    \item [(b)] $\cR_\ell$ is sensitive to large losses on $L^\infty$.
    \item[(c)] $\limsup_{x \to \infty} \tfrac{\ell(x)}{\ell(-x)} = -\infty$.
\end{enumerate}
\end{proposition}

\smallskip

\begin{example}
Let $\ell(x) = e^{\gamma x} -1$ for $x\in\R$, where $\gamma\in(0,\infty)$ denotes the parameter of absolute risk aversion. In this case, we have for every $X\in L^1$
\[
\cR_\ell(X) = \frac{1}{\gamma}\log(\E[e^{-\gamma X}]),
\]
showing that $\cR_\ell$ is the well-known entropic risk measure. It follows directly from the proposition above that $\cR_\ell$ is sensitive to large losses on $L^1$.
\end{example}


\section{Examples II: Utility functionals}
\label{sec:examples utility}

In this section we turn to reviewing a variety of utility functionals and apply the results developed so far to determine whether they are sensitive to large losses or not. Once again, we assume that, for each $\varepsilon \in (0, 1)$, there exists $A \in \cF$ with $0 < \mathbb{P}(A) \leq  \varepsilon$ and test sensitivity to large losses on the two natural domains $L^1$ and $L^\infty$, which are equipped with their canonical topological lattice structure.

\subsection{Expected utility}

We begin our discussion by studying the expected utility functional $\cE_u(X) := \E[u(X)]$, where $u:\mathbb{R} \to [-\infty,\infty)$ is a utility function, i.e., it is increasing and satisfies $u(0) = 0$ as well as 
$\limsup_{x \to \infty} \tfrac{u(x)}{x} < \infty$, cf.~Example~\ref{ex: VaR, ES, Eu}. The next theorem shows that whether expected utility is sensitive to large losses or not depends only on the ratio of the tails of the underlying utility function. More precisely, sensitivity to large losses is fulfilled whenever the loss tail asymptotically dominates the gain tail. This simple criterion can be easily checked for many concrete examples of utility functions. For instance, it follows immediately that expected utility is sensitive to large losses for many common concave utility functions including exponential, logarithmic, quadratic and power utility. As illustrated below, the result yields a simple criterion also for nonconcave utility functions.

\begin{theorem}
\label{thm:expect utility}
Let $u$ be a utility function that is negatively star-shaped (e.g.~concave) on $\R_+$. 
The following are equivalent:\footnote{Here, we agree that $\frac{u(-x)}{u(x)} =0$ if $u(-x) = u(x) =0$.
}
\begin{enumerate}
    \item[(a)] $\cE_u$ is sensitive to large losses on $L^1$.
    \item[(b)] $\cE_u$ is sensitive to large losses on $L^\infty$.
    \item[(c)] $\limsup_{x\to\infty}\frac{u(-x)}{u(x)}=-\infty$.
\end{enumerate}
\end{theorem}

\begin{remark}
An inspection of the proof shows that the implication ``(b) $\Rightarrow$ (c)'' is true without the assumption that $u$ is negatively star-shaped on $\R_+$.
\end{remark}

\smallskip

The following example shows that expected utility may fail to be sensitive to large losses even if the chosen utility function is concave and that, in the case of power-type $S$-shaped utility functions, sensitivity to large losses may or may not hold depending on the tail behaviour.

\begin{example}
(i) In general, a concave utility function does not imply sensitivity to large losses, the simplest example being given by the linear utility function $u(x)=x$ for $x\in\R$.

\smallskip

(ii) Let $\alpha,\beta\in(0,1)$. Let $u : \R \to \R$ be the $S$-shaped utility function given by
\begin{equation*}
u(x) = 
\begin{cases}
x^\alpha & \text{if } x \geq 0, \\
-(-x)^\beta & \text{if } x < 0.
\end{cases}
\end{equation*}
Theorem \ref{thm:expect utility} shows that $u$ is sensitive to large losses if and only if $\alpha < \beta$, i.e., it is less risk-seeking for losses than risk-averse for gains.
\end{example}

\subsection{Certainty equivalent}
In this section we consider different types of certainty equivalents, which qualify as special utility functionals based on the expected utility functional from the previous section.
\subsubsection{The classical certainty equivalent}

For a utility function $u$, we define its generalized inverse $u^{-1}:  u(\R) \to [-\infty, \infty)$ by
\begin{equation*}
    u^{-1}(y) := \inf\{ x \in \R : u(x) \geq y \}.
\end{equation*}
The \emph{(classical) certainty equivalent} of a random variable $X \in L^1$ is the least sure payoff $\cC_u(X)$ that leaves the decision maker at least as satisfied as with $X$, i.e.,
\begin{equation*}
    \cC_u(X):= u^{-1}(\mathbb{E}[u(X)]).
\end{equation*}
The functional $\cC_u$ is a utility functional in our general sense if, and only if, $u(x) < 0$ for every $x \in (-\infty, 0)$. We can fully characterize when $\cC_u$ is sensitive to large losses for a large class of utility functions. Not surprisingly, the characterization is identical to the one obtained for expected utility. This is because, as stated in the proof, the certainty equivalent is sensitive to large losses precisely when the associated expected utility is.

\begin{proposition}
\label{prop:Cu SLL}
Let $u$ be a utility function that is negatively star-shaped (e.g.~concave) on $\R_+$, left-continuous at $0$ and $u(x) < 0$ for every $x \in (-\infty, 0)$. The following statements are equivalent:
\begin{enumerate}
    \item[(a)] $\cC_u$ is sensitive to large losses on $L^1$.
    \item[(b)] $\cC_u$ is sensitive to large losses on $L^\infty$.
    \item[(c)] $\limsup_{x\to\infty}\frac{u(-x)}{u(x)}=-\infty$.
\end{enumerate}
\end{proposition}

\begin{remark}
Proposition \ref{prop:Cu SLL} may fail without the assumption that $u$ is left-continuous at $0$. For example if $u(x)=-1$ for $x < 0$ and $u(x) = 0$ for $x > 0$, then the limit condition in point (c) holds but, for every $X\in L^1$ with $\mathbb{P}(X < 0) < 1$, we have $\cC_u(X) = 0$.
\end{remark}

\subsubsection{The u-mean certainty equivalent}

Based on the ``principle of zero utility'', for a strictly increasing utility function $u$, B{\"u}hlmann \cite{buhlmann2007mathematical} defined the \emph{$u$-mean certainty equivalent} $\cM_u(X)$ of a random variable $X \in L^1$ as the unique solution to the equation 
\begin{equation*}
    \mathbb{E}[u(X-\cM_u(X))] = 0.
\end{equation*}
A natural extension of this definition (see, for example, \cite[Section 5.2]{ben2007old}) that can be applied to any utility function $u$ is 
\begin{equation*}
    \cM_u(X):= \sup \{ m \in \R \,; \ \mathbb{E}[u(X-m)] \geq 0 \}.
\end{equation*}
This is a utility functional in our generalised sense if, and only if, $u(x) < 0$ for every $x \in (-\infty, 0)$. Contrary to the case of classical certainty equivalents, it is not clear to what extent sensitivity to large losses of $u$-mean certainty equivalents and $u$-based expected utility are related to each other. The next result, whose proof is somewhat delicate, shows that $\cM_u$ is sensitive to large losses precisely when the associated expected utility $\cE_u$ is, making the verification of this property solely based on the tail behaviour of the utility function at both ends of its spectrum.

\begin{proposition}
\label{prop:Mu SLL}
Let $u$ be a utility function that is negatively star-shaped (e.g.\ concave) on $\mathbb{R}_+$ and assume $u(x) < 0$ for every $x \in (-\infty,0)$. The following are equivalent:
\begin{enumerate}
    \item[(a)] $\cM_u$ is sensitive to large losses on $L^1$.
    \item[(b)] $\cM_u$ is sensitive to large losses on $L^\infty$.
    \item[(c)] $\limsup_{x \to \infty} \tfrac{u(-x)}{u(x)} = -\infty$.
\end{enumerate}
\end{proposition}

\subsubsection{The optimized certainty equivalent}

In \cite{ben1986expected}, another certainty equivalent was introduced, called the \emph{optimized certainty equivalent}, which was further studied by the same authors in \cite{ben2007old}. Here, we consider the same concept but in far greater generality since we do not require $u$ to be concave (or even star-shaped). Given a utility function $u$ such that $u(x) \leq x$ for every $x \in \R$, the optimized certainty equivalent of a random variable $X\in L^1$ is defined as
\begin{equation*}
\OCE_{u}(X) := \sup_{\eta \in \R} \{ \eta + \mathbb{E}[u(X - \eta)] \}.
\end{equation*}
This is a cash-additive utility functional in the sense of our general definition (the property $u(x) \leq x$ for every $x \in \R$ ensures $\OCE_u$ is well-defined and normalised). Moreover, it is “sandwiched” between the expectation and expected utility, i.e., for $X \in L^1$,
\begin{equation*}
    \mathbb{E}[X] = \sup_{\eta \in \R} \{ \eta + \mathbb{E}[X - \eta] \} \geq \OCE_u(X) \geq 0 + \mathbb{E}[u(X - 0)] = \mathcal{E}_u(X).
\end{equation*}
The interpretation of the optimized certainty equivalent is as follows: a decision maker expects a future uncertain income $X$, and can consume part of $X$ at present. If they choose to consume $\eta$, then the resulting present value is $\eta + \mathbb{E}[u(X-\eta)]$. Thus, the optimized certainty equivalent is the result of an optimal allocation of $X$ between present and future consumption. The next result establishes exactly when it is sensitive to large losses. Interestingly enough, the asymptotic condition on the utility function that turns out to be equivalent to sensitivity to large losses is different from the one we have encountered thus far. In fact, as illustrated below, it is strictly stronger. This implies that sensitivity to large losses of the optimized certainty equivalent is sufficient, but generally not necessary, for sensitivity to large losses of expected utility or of the classical and $u$-mean certainty equivalent.

\begin{theorem}
\label{theo: OCE}
Let $u$ be a utility function that satisfies $u(x) \leq x$ for $x \in \R$. The following are equivalent:
\begin{enumerate}
    \item[(a)] $\OCE_u$ is sensitive to large losses on $L^1$.
        \item[(b)] $\OCE_u$ is sensitive to large losses on $L^\infty$.
    \item[(c)] $\liminf_{x\to-\infty}\tfrac{u(x)}{x} = \infty$ and $\limsup_{x\to\infty}\tfrac{u(x)}{x}=0$.
\end{enumerate}
\end{theorem}

\begin{remark}
(a) The asymptotic condition in (c) always implies that $\limsup_{x\to\infty}\tfrac{u(-x)}{u(x)} = -\infty$, which was proved to be equivalent to sensitivity to large losses of expected utility and of both the classical and the $u$-mean certainty equivalent. To see this, it suffices to observe that (c) yields
\[
\limsup_{x\to\infty}\tfrac{u(-x)}{u(x)} \le \limsup_{x\to\infty}\tfrac{u(-x)}{x}\limsup_{x\to\infty}\tfrac{x}{u(x)} = -\frac{\liminf_{x\to-\infty}\tfrac{u(x)}{x}}{\liminf_{x\to\infty}\tfrac{u(x)}{x}} = -\tfrac{\infty}{0} = -\infty.
\]

(b) Let $u$ be the utility function given by 
\begin{equation*}
u(x) = \begin{cases}
x^\alpha, & \text{if } x > 1, \\
x, & \text{if } 0 \leq x \leq 1, \\
\beta x, & \text{if } x < 0,
\end{cases}
\end{equation*}
where $0 < \alpha < 1 \leq \beta$. Then $\OCE_u$ fails to be sensitive to large losses whereas $\cE_u$ is sensitive to large losses. It also shows that (c) is generally not necessary for $\limsup_{x\to\infty}\tfrac{u(-x)}{u(x)} = -\infty$ to hold.
\end{remark}


\section{Conclusion and further research}
\label{sec:conclusion}

In this paper we have revisited the classical topics of tail risk and loss aversion by introducing the new property of {\em sensitivity to large losses} for risk and utility functionals. This property axiomatizes the requirement that a sufficiently large rescaling of a financial position that exposes an agent to losses should be marked by a positive risk or a negative utility. We have established a number of necessary and sufficient conditions for sensitivity to large losses, which highlight the broad applicability and implementability of this new axiom. Our examples cover standard and more recent risk functionals including (loss) Value at Risk and (adjusted) Expected Shortfall as well as utility functionals like expected utility and (optimized) certainty equivalents associated with both concave and nonconcave (including star- and $S$-shaped) utility functions.

\smallskip

Complementing the general results established in this paper, we envision two main directions of future research. Both of them have to do with the choice of the underlying domain of financial positions and yield important financial applications. In particular, when the chosen risk functionals are Value at Risk and Expected Shortfall, this has relevant implications for financial regulation.

\smallskip

First, it would be interesting to investigate in a more comprehensive way the dependence of sensitivity to large losses on the choice of the underlying domain. We already know that important risk or utility functionals such as Expected Shortfall or expected utility functionals may or may not be sensitive to large losses based on that choice. It is therefore a natural follow-up question to find the {\em maximal domain} of sensitivity to large losses for a given functional and construct an associated \emph{index of sensitivity to large losses} to allow for better comparison across different functionals. Here, the case of expected utility functionals seems particularly tractable: the index could be defined in terms of the asymptotic behaviour of the utility function as described in Theorem \ref{thm:expect utility}.

\smallskip

Second, it would be worthwhile to study sensitivity to large losses on {\em specific domains} of financial positions. On the one hand, this is interesting in itself, and on the other hand it reveals some deep connections between sensitivity to large losses and some fundamental problems in finance. For example, consider a one-period arbitrage-free frictionless financial market where a finite number of assets with terminal payoffs $S_1,\dots,S_N\in\cX$ are traded. Denote by $\cM$ the linear space spanned by the assets' payoffs and by $\pi:\cM\to\R$ a linear pricing functional. The set
\[
\cS=\{X\in\cM \,; \ \pi(X)\le0\}
\]
consists of all replicable payoffs that can be purchased at zero cost. Expected Shortfall is sensitive to large losses on $\cS$ precisely when every nonzero replicable payoff with zero cost has a strictly positive risk, i.e., $\ES_\alpha(X)>0$ for every nonzero $X\in\cS$. This condition can be interpreted as the absence of a weak form of arbitrage opportunity, which could be termed {\em good deal induced by Expected Shortfall} in the language of, e.g., \cite{ArducaMunari2023,BaesKochMunari2020,CernyHodges2002,JaschkeKuechler2001}. In a similar spirit, let $N=d+1$ and assume the first asset is risk-free with relative return $r\in(-1,\infty)$. The remaining $d$ assets are risky and their relative returns are denoted by $R_1,\dots,R_d\in\cX$. The set
\[
\cS = \left\{ X_\pi:=\sum_{i=1}^d\pi_i(R_i-r) \,; \ \pi=(\pi_1,\dots,\pi_d)\in\R^d \right\}
\]
consists of all excess returns over the risk-free rate. It turns out that Expected Shortfall is sensitive to large losses on $\cS$ if, and only if, there is no portfolio sequence $(\pi_n)\subset\R^d$ such that $\E[X_{\pi_n}]\uparrow\infty$ and $\ES_\alpha(X_{\pi_n})\le 0$ for every $n\in\N$. This can also be interpreted as the absence of a weak form of arbitrage opportunity, which is termed {\em $\rho$-arbitrage}, more specifically $\ES_\alpha$-arbitrage, in \cite{herdegen2022mean,HerdegenKhan2024}. 

\smallskip
These examples show that a thorough study of sensitivity to large losses can lead to the development of new, perhaps surprising, results in the field of (arbitrage) pricing and (mean-risk) portfolio selection. To this effect, a key step is to target a dual description of sensitivity to large losses where the dual variables play the role of special stochastic discount factors or risk-neutral probabilities thereby creating the link with the existing literature on good deal pricing and mean-risk portfolio selection.

\appendix

\section{Auxiliary results and proofs}
\label{sec:appendix}

\begin{proposition}
\label{prop: appendix}
Let $\cH :\cX \to(-\infty,\infty]$ be positively star-shaped. Then for all $X \in \cX$ and $\lambda_1,\lambda_2\in(0,\infty)$ with $\lambda_1<\lambda_2$ we have $\frac{\cH(\lambda_1 X)}{\lambda_1}\le\frac{\cH(\lambda_2 X)}{\lambda_2}$. In particular,
\[
\sup_{\lambda\in(0,\infty)}\tfrac{\cH(\lambda X)}{\lambda}=\sup_{\lambda\in(1,\infty)}\tfrac{\cH(\lambda X)}{\lambda}=\lim_{\lambda\to\infty}\tfrac{\cH(\lambda X)}{\lambda}.
\]
\end{proposition}
\begin{proof}
Note that, by positive star-shapedness, for all $X \in \cH$ and $\lambda_1,\lambda_2\in(0,\infty)$ with $\lambda_1<\lambda_2$,
\[
\cH(\lambda_2 X) = \cH(\tfrac{\lambda_2}{\lambda_1}\lambda_1 X) \ge \tfrac{\lambda_2}{\lambda_1}\cH(\lambda_1 X).
\]
This yields the first statement, which in turn yields the other statement.
\end{proof}

\smallskip

{\bf Proof of Proposition~\ref{prop: first link with recession map}}. The equivalence between (a) and (b) follows from Proposition~\ref{prop: positively homogeneous} by positive homogeneity of $\cR^\infty$. To establish equivalence between (b) and (c), take any $X\in\cS$ with $\probp(X<0)>0$ and note that $\cR^\infty(X)>0$ if, and only if, $\cR(\lambda X)>0$ for some $\lambda\in(0,\infty)$.\hfill\qed

\medskip

{\bf Proof of Theorem~\ref{theo: recession functional}}. It is clear that (b) implies (c), and it follows from Proposition~\ref{prop: first link with recession map} that (c) and (d) are equivalent. By Proposition~\ref{prop: appendix}, for any $X\in\cX$ and $\lambda_1,\lambda_2\in(0,\infty)$ with $\lambda_1<\lambda_2$,
\begin{equation}
\label{eq: star-shaped}
\tfrac{\cR(\lambda_1X)}{\lambda_1} \le \tfrac{\cR(\lambda_2X)}{\lambda_2}.
\end{equation}
To show that (c) implies (a), let $X\in\cS$ with $\probp(X<0)>0$. By (c), there exists $\lambda_X\in(0,\infty)$ such that $\cR(\lambda_X X)>0$. Then \eqref{eq: star-shaped} gives $\cR(\lambda X)>0$ for every $\lambda\in(\lambda_1,\infty)$. Finally, to prove that (a) implies (b), assume that $\cR$ is sensitive to large losses on $\cS$. By assumption, for every $X\in\cS$ with $\probp(X<0)>0$, there exists $\lambda_1\in(0,\infty)$ such that $\cR(\lambda_1X)>0$. As a consequence of \eqref{eq: star-shaped},
\[
\sup_{\lambda\in(0,\infty)}\cR(\lambda X) \ge \sup_{\lambda\in(\lambda_1,\infty)}\big\{\tfrac{\cR(\lambda X)}{\lambda}\lambda\big\} \ge \sup_{\lambda\in(\lambda_1,\infty)}\big\{\tfrac{\cR(\lambda_1 X)}{\lambda_1}\lambda\big\} = \infty.
\]
This concludes the proof of the equivalence. As a last step, assume that (a) holds. In view of the equivalence between (a) and (d), to establish (e) it suffices to prove that, for fixed $X\in\cX$,
\begin{equation}
\label{eq: rho hat def}
\cR^\infty(X) \le \widehat{\cR}(X) := \sup_{Y\in\dom(\cR)}\{\cR(X+Y)-\cR(Y)\}.
\end{equation}
To this effect, assume with no loss of generality that $\widehat{\cR}(X)<\infty$. We show by induction that
\begin{equation}
\label{eq: rho hat}
\cR(nX)\le n\widehat{\cR}(X), \ \ \ \forall n\in\N.   
\end{equation}
The statement is clear when $n=1$ because $0\in\dom(\cR)$ and $\cR(X)\le\cR(X+0)-\cR(0)\le\widehat{\cR}(X)$. If it holds for some $n\in\N$, then $nX\in\dom(\cR)$ and thus
\[
\cR((n+1)X)=\cR(X+nX)\le\widehat{\cR}(X)+\cR(nX)\le(n+1)\widehat{\cR}(X).
\]
This yields~\eqref{eq: rho hat}. Take now an arbitrary $\lambda\in(0,\infty)$. We find $n\in\N$ and $\xi\in(0,1]$ such that $\lambda=\xi n$. It then follows from star-shapedness of $\cR$ and~\eqref{eq: rho hat} that
\[
\tfrac{\cR(\lambda X)}{\lambda} = \tfrac{\cR(\xi n X)}{\xi n} \le \tfrac{\xi\cR(nX)}{\xi n} \le \widehat{\cR}(X).
\]
Taking a supremum over $\lambda$ delivers~\eqref{eq: rho hat def} and concludes the proof.\hfill\qed

\medskip

{\bf Proof of Theorem~\ref{theo: convex}}. As every convex risk functional is star-shaped, we only have to prove that (e) implies (d) by Theorem \ref{theo: recession functional}. To this end, it suffices to show that, for every $X\in\cX$,
\begin{equation}
\label{eq: rho hat again}
\cR^\infty(X) \ge \widehat{\cR}(X) := \sup_{Y\in\dom(\cR)}\{\cR(X+Y)-\cR(Y)\}.
\end{equation}
To this effect, denote by $\cX^\ast$ the topological dual of $\cX$ and define the conjugate function $\cR^\ast:\cX^\ast\to(-\infty,\infty]$ by setting
\[
\cR^\ast(\psi) := \sup_{X\in\cX}\{\psi(X)-\cR(X)\}.
\]
Now, fix an arbitrary $X\in\cX$. As $\cR$ is convex and lower semicontinuous, it follows from the classical Fenchel-Moreau-Rockafellar representation, see, e.g., \cite[Theorem 2.3.3]{Zalinescu}, that
\begin{equation}
\label{eq: dual representation}
\cR(X) = \sup_{\psi\in\dom(\cR^\ast)}\{\psi(X)-\cR^\ast(\psi)\},
\end{equation}
where $\dom(\cR^\ast):=\{\psi\in\cX^\ast \,; \ \cR^\ast(\psi)<\infty\}$. Note that, for every $\psi\in\dom(\cR^\ast)$, we have $\cR^\ast(\psi)\ge\psi(0)-\cR(0)\ge0$. As a result, the dual representation in~\eqref{eq: dual representation} gives
\[
\cR^\infty(X) = \sup_{\lambda\in(0,\infty)}\tfrac{1}{\lambda}\sup_{\psi\in\dom(\cR^\ast)}\{\psi(\lambda X)-\cR^\ast(\psi)\} = \sup_{\psi\in\dom(\cR^\ast)}\sup_{\lambda\in(0,\infty)}\{\psi(X)-\tfrac{\cR^\ast(\psi)}{\lambda}\} = \sup_{\psi\in\dom(\cR^\ast)}\psi(X).
\]
Now, take any $Y\in\dom(\cR)$. Another application of~\eqref{eq: dual representation} yields
\begin{align*}
\cR^\infty(X)+\cR(Y) &= \sup_{\psi\in\dom(\cR^\ast)}\psi(X)+\sup_{\psi\in\dom(\cR^\ast)}\{\psi(Y)-\cR^\ast(\psi)\} \\ &\ge \sup_{\psi\in\dom(\cR^\ast)}\{\psi(X+Y)-\cR^\ast(\psi)\} = \cR(X+Y).
\end{align*}
Reshuffling and taking a supremum over $Y$ delivers~\eqref{eq: rho hat again} and concludes the proof.\hfill\qed

\medskip

{\bf Proof of Theorem~\ref{theo: cash-additive}}. If (b) holds, then $\cR^\infty$ is sensitive to large losses on $\cS$ and, hence, (a) holds by Theorem~\ref{theo: recession functional}. To prove that (a) implies (b), assume that $\cR$ is sensitive to large losses on $\cS$. As $\cR^\infty$ inherits monotonicity from $\cR$ and $\cR^\infty(0)=0$, we deduce that $\cX_+\subset\cA_{\cR^\infty}$. Conversely, again by Theorem~\ref{theo: recession functional}, we have $\cA_{\cR^\infty}\cap\cS\subset\cX_+$. As a result, $\cA_{\cR^\infty}\cap\cS=\cX_+\cap\cS$. As $\cR^\infty$ is cash-additive, for every $X\in\cS$,
\begin{align*}
\cR^\infty(X) &=
\inf\{m\in\R \,; \ X+m\in\cA_{\cR^\infty}\} = \inf\{m\in\R \,; \ X+m\in\cA_{\cR^\infty}\cap\cS\} \\
&= \inf\{m\in\R \,; \ X+m\in\cX_+\cap\cS\} = \inf\{m\in\R \,; \ X+m\in\cX_+\} = \esssup(-X),
\end{align*}
where we used that $\cS+\R\subset\cS$. This concludes the proof of the equivalence.\hfill\qed

\medskip

\begin{proposition}
\label{prop: rho A risk measure}
For a risk functional $\cR$ the following statements are equivalent:
\begin{enumerate}
    \item[(a)] $\cR_{\cA_\cR}$ is a risk functional.
    \item[(b)] $\cR(m)>0$ for $m\in(-\infty,0)$ and for any $X\in\cX$ we have $\cR(X+m)>0$ for some $m\in\R$.
\end{enumerate}
\end{proposition}
\begin{proof}
First, assume that (b) holds. As $\cA_\cR$ is monotone and contains $0$ by monotonicity and normalization of $\cR$, we easily deduce that $\cR_{\cA_\cR}+\cX_+\subset\cR_{\cA_\cR}$ and $\cR_{\cA_\cR}(0)\le0$. Since $\cA_\cR\cap(-\infty,0)=\emptyset$ by assumption, we in fact have $\cR_{\cA_\cR}(0)=0$. Finally, again by assumption, for every $X\in\cX$ there exists $m\in\R$ such that $\cR(X+m)>0$, whence $\cR_{\cA_\cR}(X)\ge m>-\infty$ by monotonicity of $\cR$. This shows that (a) holds. Conversely, suppose that (a) holds. Then, $\cR_{\cA_\cR}(0)=0$, showing that $\cA_\cR\cap(-\infty,0)=\emptyset$. Furthermore, for every $X\in\cX$ we have $\cR_{\cA_\cR}(X)>-\infty$, implying that one must find $m\in\R$ such that $X+m\notin\cA_\cR$. As a result, (b) holds.
\end{proof}

\medskip

{\bf Proof of Proposition~\ref{prop: sensitivity via acceptability preliminary}}. Take an arbitrary $X\in\cS$ with $\probp(X<0)>0$. If $\cR_{\cA_\cR}$ is sensitive to large losses on $\cS$, then we find $\lambda_X\in(0,\infty)$ such that each $\lambda\in(\lambda_X,\infty)$ satisfies $\cR_{\cA_\cR}(\lambda X)>0$, which implies $\lambda X\notin\cA_\cR$, whence $\cR(\lambda X)>0$. Then, $\cR$ is also sensitive to large losses on $\cS$.\hfill\qed

\medskip

{\bf Proof of Theorem~\ref{the: sensitivity via acceptability v2}}. By Proposition \ref{prop: sensitivity via acceptability preliminary}, we only have to prove that (a) implies (b). To this effect, assume that $\cR$ is regular. Take any $X\in\cS$ with $\probp(X<0)>0$. By assumption, we find $\lambda_X\in(0,\infty)$ such that $\cR(\lambda X)>0$ for every $\lambda\in(\lambda_X,\infty)$. If there existed $\lambda\in(\lambda_X,\infty)$ for which $\cR_{\cA_\cR}(\lambda X)\le0$, then we would find a sequence $(m_n)\subset\R$ and $m\in(-\infty,0]$ satisfying $m_n\to m$ and $\lambda X+m_n\in\cA_\cR$ for every $n\in\N$. However, monotonicity and lower regularity would entail
\[
0 < \cR(\lambda X) \le \cR(\lambda X+m) \le \liminf_{n\to\infty}\cR(\lambda X+m_n) \le 0, 
\]
which is clearly impossible. As a consequence, $\cR_{\cA_\cR}(\lambda X)>0$ for every $\lambda\in(\lambda_X,\infty)$, showing that $\cR_{\cA_\cR}$ is sensitive to large losses on $\cS$ and concluding the proof.\hfill\qed

\medskip

In the proof of Theorem~\ref{thm:sensitivity via acceptability v3} we use that the recession functional of a risk/utility functional defined through a set $\cA\subset\cX$ coincides with the risk/utility functional defined through its recession set
\[
\cA^\infty := \{X\in\cX \,; \ \lambda X\in\cA, \ \forall \lambda\in(0,\infty)\}.
\]

\begin{lemma}
\label{lem: recession functionals and cones}
Let $\cA\subset\cX$. For every $X\in\cX$, we have $(\cR_\cA)^\infty(X) = \cR_{\cA^\infty}(X)$.
\end{lemma}
\begin{proof}
Fix $X\in\cX$. We first show $(\cR_\cA)^\infty(X) \leq \cR_{\cA^\infty}(X)$.  To this end, take any $m\in\R$ with $X+m\in\cA^\infty$. Then, for every $\lambda\in(0,\infty)$, we have $\lambda X+\lambda m\in\cA$ and, thus, $\cR_\cA(\lambda X)\le\lambda m$, showing that $(\cR_\cA)^\infty(X)\le m$. Conversely, to prove $(\cR_\cA)^\infty(X) \geq \cR_{\cA^\infty}(X)$, take any $m\in\R$ with $m>(\cR_\cA)^\infty(X)$. For every $\lambda\in(0,\infty)$ we see that $\cR_\cA(\lambda X)<\lambda m$, which implies that $\lambda X+\lambda m\in\cA$, whence $X+m\in\cA^\infty$. This shows that $\cR_{\cA^\infty}(X)\le m$.
\end{proof}

\medskip

{\bf Proof of Theorem~\ref{thm:sensitivity via acceptability v3}}. Once again, by Proposition \ref{prop: sensitivity via acceptability preliminary}, we only have to prove that (a) implies (b). To this end, assume that $\cR$ is sensitive to large losses on $\cS$. As a preliminary remark, note that $\cR_{\cA_{\cR^\infty}}$ is a risk functional by Proposition~\ref{prop: rho A risk measure} because $\cR^\infty$ is a risk functional, implying that $\cA_{\cR^\infty}$ is monotone and $\cA_{\cR^\infty}\cap\R=[0,\infty)$, and $\cR^\infty\ge\cR$, implying that $\cA_{\cR^\infty}\subset\cA_\cR$ and, hence, $(\cA_{\cR^\infty}-X)\cap\R\neq\R$ for every $X\in\cX$. Now, it follows from Theorem~\ref{theo: recession functional} that $\cR^\infty$ is sensitive to losses on $\cS$. As a result, $\cA_{\cR^\infty}\cap\cS=\cX_+\cap\cS$ and, for every $X\in\cS$,
\begin{align*}
\cR_{\cA_{\cR^\infty}}(X) &=
\inf\{m\in\R \,; \ X+m\in\cA_{\cR^\infty}\} = \inf\{m\in\R \,; \ X+m\in\cA_{\cR^\infty}\cap\cS\} \\
&= \inf\{m\in\R \,; \ X+m\in\cX_+\cap\cS\} = \inf\{m\in\R \,; \ X+m\in\cX_+\} = \esssup(-X),
\end{align*}
where we used that $\cS+\R=\cS$. This shows that $\cR_{\cA_{\cR^\infty}}$ is also sensitive to losses on $\cS$. Note that $(\cA_\cR)^\infty=\cA_{\cR^\infty}$ because, for each $X\in\cX$, we have $X\in(\cA_\cR)^\infty$ if, and only if, $\cR(\lambda X)\le0$ for every $\lambda\in(0,\infty)$. Hence, we have $\cR_{\cA_{\cR^\infty}}=\cR_{(\cA_\cR)^\infty}=(\cR_{\cA_\cR})^\infty$ by Lemma~\ref{lem: recession functionals and cones}. We deduce that $(\cR_{\cA_\cR})^\infty$ is sensitive to losses on $\cS$. Now, it is fairly easy to verify that $\cR_{\cA_\cR}$ inherits star-shapedness from $\cR$ because, for all $X\in\cX$ and $\lambda\in[1,\infty)$ and for each $m\in\R$ with $\lambda X+m\in\cA_\cR$, we have $\lambda\cR(X+\frac{m}{\lambda}) \le \cR(\lambda(X+\frac{m}{\lambda})) \le 0$ by star-shapedness of $\cR$, showing that $\cR_{\cA_\cR}(X) \le \frac{m}{\lambda}$ and, by taking the infimum of such $m$'s, delivering $\lambda\cR_{\cA_\cR}(X) \le \cR_{\cA_\cR}(\lambda X)$ as claimed. It then follows from Theorem~\ref{theo: recession functional} that $\cR_{\cA_\cR}$ is sensitive to large losses on $\cS$, concluding the proof.\hfill\qed

\medskip

{\bf Proof of Proposition~\ref{prop:sens large exp/pure:1}}. We only proof part (a), the proof of part (b) is trivial. Take any nonzero $X\in\cX$ with $\E[X] \leq 0$. Set $Y :=  X - \E[X]$. Then $X \leq Y$ and $Y$ is nonconstant with $\E[Y] = 0$. Hence by monotonicity of $\cR^\infty$ and the hypothesis
$\cR^\infty(X)\ge \cR^\infty(Y) > 0$. Now the claim follow from Theorem \ref{theo: recession functional}.

\medskip

{\bf Proof of Theorem~\ref{thm:sens large exp/pure:2}}. 
In light of Proposition~\ref{prop:sens large exp/pure:1}, it suffices to show that if $\cR$ is sensitive to large expected losses, then $\cR^\infty(X) > \E[-X]$ for all non-constant $X$. So let $X \in \cX$ be non-constant. Set $Y := X - \E[X]$. By cash-invariance, it suffices to show that $\cR^\infty(Y) > \E[-Y] = 0$. But this follows directly from \ref{theo: recession functional} and the fact that $\cR$ is sensitive to large losses on $\cX_e = \{X \in \cX; \E[X] \leq 0\}$.

\medskip

{\bf Proof of Proposition~\ref{prop: loss concentrations}}. Let $X\in\cX$ and $A\in\cF$ with $\probp(A)>0$. Take $a,b\in(0,\infty)$ large enough that $B:=A\cap\{X\le a\}$ and $Y:=-b\ind_B+\max\{X,0\}\ind_{B^c}$ satisfy $\probp(B)>0$ and $\E[Y]\le0$. In particular, $Y\in\cX_e$. By sensitivity to large losses on $\cX_e$, we find $\lambda\in(1,\infty)$ such that $\cR(\lambda Y)>0$. Letting $\xi := a + \lambda b$ we obtain $X-\xi\ind_A\le\lambda Y$. As a consequence of monotonicity, $\cR(X-\xi\ind_A)\ge\cR(\lambda Y)>0$, proving that $\cR$ is sensitive to loss concentrations.\hfill\qed

\medskip

{\bf Proof of Proposition~\ref{prop: necessary for sensitivity to loss concentrations}}. Let $X \in \cX_{-} \setminus \{0\}$. By definition, there exists $\varepsilon\in(0,\infty)$ such that $\mathbb{P}(X \leq -\varepsilon) > 0$, and note that $\lambda X \leq -\lambda \varepsilon \ind_{\{ X \leq -\varepsilon \}}$ for every $\lambda \in (0,\infty)$. Now since $\cR$ is sensitive to loss concentrations, there exists $\gamma \in (0,\infty)$ such that $\cR(0 - \gamma \ind_{\{ X \leq -\varepsilon \}}) > 0$. By the monotonicity of $\cR$ it follows that $\cR(\lambda X) > 0$ for every $\lambda \in (\gamma/\varepsilon, \infty)$. Therefore, $\cR$ is sensitive to large pure losses.\hfill\qed

\medskip

{\bf Proof of Theorem~\ref{thm: loss concentrations part two}}. In light of Proposition \ref{prop: necessary for sensitivity to loss concentrations}, it suffices to that sensitivity to large pure losses implies sensitivity to loss concentrations.

Take $A\in\cF$ with $\probp(A)>0$ and note that $\cR(-\lambda\ind_A)>0$ for some $\lambda\in(0,\infty)$ by sensitivity to large pure losses. It follows from lower semicontinuity and star-shapedness that, for every $X\in\cX$,
\[
0 < \cR(-\lambda\ind_A) \le \liminf_{n\to\infty}\cR\big(\tfrac{X}{n}-\lambda\ind_A\big) = \liminf_{n\to\infty}\cR\big(\tfrac{1}{n}(X-n\lambda\ind_A)\big) \le \liminf_{n\to\infty}\big[\tfrac{1}{n}\cR(X-n\lambda\ind_A)\big].
\]
This implies that $\cR(X-n\lambda\ind_A)>0$ for some $n\in\N$, showing that $\cR$ is sensitive to loss concentrations.\hfill\qed



\medskip

{\bf Proof of Proposition~\ref{prop: loss var}}. For every $X\in L^1$ the recession functional of $\LVaR_\alpha$ satisfies
\[
\LVaR_\alpha^\infty(X) = \sup_{\lambda\in(0,\infty)}\{\tfrac{1}{\lambda}\sup_{l\in(-\infty,0]}\{\VaR_{\alpha(l)}(\lambda X)+l\}\} = \sup_{l\in(-\infty,0]}\{\VaR_{\alpha(l)}(X)+\sup_{\lambda\in(0,\infty)}\tfrac{l}{\lambda}\}
\]
by positive homogeneity of $\VaR_{\alpha(l)}$'s. Clearly, $\sup_{\lambda\in(0,\infty)}\tfrac{l}{\lambda}=0$. Hence, for every $X\in L^1$,
\[
\LVaR_\alpha^\infty(X) = 
\begin{cases}
\VaR_{\underline{\alpha}}(X) & \mbox{if} \ \underline{\alpha}>0,\\
\esssup(-X) & \mbox{if} \ \underline{\alpha}=0.
\end{cases}
\]
With this in mind, ``(c) $\Rightarrow$ (a)'' follows from Theorem \ref{theo: cash-additive} whilst ``(b) $\Rightarrow$ (c)'' follows from Example~\ref{ex: var and es}(a). The implication ``(a) $\Rightarrow$ (b)'' is trivial. \hfill\qed

\medskip

{\bf Proof of Proposition~\ref{prop: adjusted es}}. For every $X\in L^1$ the recession functional of $\ES^g$ satisfies
\[
(\ES^g)^\infty(X) = \sup_{\lambda\in(0,\infty)}\{\tfrac{1}{\lambda}\sup_{\alpha\in(0,1]}\{\ES_\alpha(\lambda X)-g(\alpha)\}\} = \sup_{\alpha\in(0,1]}\{\ES_\alpha(X)-\inf_{\lambda\in(0,\infty)}\tfrac{g(\alpha)}{\lambda}\},
\]
where we used positive homogeneity of $\ES_\alpha$'s. Clearly, for every $\alpha\in(0,1]$,
\[
\inf_{\lambda\in(0,\infty)}\tfrac{g(\alpha)}{\lambda}=
\begin{cases}
0 & \mbox{if} \ g(\alpha)<\infty,\\
\infty & \mbox{if} \ g(\alpha)=\infty.
\end{cases}
\]
If $g$ is finite everywhere, we deduce that, for every $X\in L^1$,
\[
(\ES^g)^\infty(X) = \sup_{\alpha\in(0,1]}\ES_\alpha(X) = \esssup(-X).
\]
Thus, ``(c) $\Rightarrow$ (a)'' follows from Theorem \ref{theo: cash-additive}. Otherwise, if $g$ is not finite everywhere, set $p:=\sup\{\alpha\in(0,1] \,; \ g(\alpha)=\infty\} \in (0,1]$. Then, regardless of whether $g(p)$ is finite or not, for every $X\in L^\infty$,
\[
(\ES^g)^\infty(X) = \sup_{\alpha\in[p,1]}\ES_\alpha(X) = \ES_p(X).
\]
Now ``(b) $\Rightarrow$ (c)'' follows from Example~\ref{ex: var and es}(b). Finally, ``(a) $\Rightarrow$ (b)'' is trivial. \hfill\qed

\medskip

{\bf Proof of Proposition~\ref{prop:SR SLL}}. This result follows from Proposition \ref{prop:Mu SLL} since $\mathcal{R}_\ell(X) = -\cM_u(X)$ for every $X \in L^1$, and $\limsup_{x\to\infty}\frac{u(-x)}{u(x)}=\limsup_{x\to\infty}\frac{\ell(x)}{\ell(-x)}$, where $u(x):=-\ell(-x)$ for $x \in \R$.\hfill\qed

\medskip

{\bf Proof of Theorem~\ref{thm:expect utility}}. 
“(a) $\Rightarrow$ (b)”: This is trivial.
“(b) $\Rightarrow$ (c)”: We argue by contraposition. Suppose that $\ell :=\limsup_{x\to\infty}\frac{u(-x)}{u(x)}>-\infty$ and note that $\ell \le0$ by normalisation of $u$. Then there exists
a sequence $(\lambda_n)_{n \in \NN}$ with $\lim_{n \to \infty} \lambda_n = +\infty$
such that $u(-\lambda_n)\ge(\ell-1)u(\lambda_n)$ for every $n \in \NN$. Take $A\in\cF$ with $0 < \probp(A) \leq \frac{1}{2-\ell}\in(0,1)$ and define $X:=\ind_{A^c}-\ind_{A}\in L^1$. Then for every $n \in \NN$,
\begin{align*}
\cE_u(\lambda_n X) &= \probp(A^c)u(\lambda_n)+\probp(A)u(-\lambda_n) \ge (1 - \probp(A))u(\lambda_n)+ \probp(A)(\ell-1)u(\lambda_n) \\
&\ge \frac{1-\ell}{2-\ell}u(\lambda_n)+\frac{1}{2-\ell}(\ell-1)u(\lambda_n) = 0,
\end{align*}
whence $\cE_u$ is not sensitive to large losses on $L^1$.

“(c) $\Rightarrow$ (a)”: Let $X\in L^1$ with $\probp(X<0)>0$. Let $\varepsilon > 0$ be such that $F_{-\varepsilon} := \probp(X\le-\varepsilon)>0$. For $n \in \NN$, set $p_{\varepsilon n} := \probp((n-1)\varepsilon  \leq X < n \varepsilon )$ and note that
\begin{equation*}
  \sum_{n =1}^\infty n  p_{\varepsilon n} \leq \sum_{n =1}^\infty \E\left[\left(\frac{X}{\varepsilon}+1\right)\ind_{\{X \in [(n-1)\varepsilon, n\varepsilon)\}}\right]\leq \frac{1}{\varepsilon} \mathbb{E}[X^+] + 1 < \infty. 
\end{equation*}
Then for $\lambda > 0$, using that $u$ is negatively-star shaped on $\R_+$, we obtain
\begin{align*}
\cE_u(\lambda X) &\le u(-\lambda\varepsilon) F_{-\varepsilon}+ \sum_{n = 1}^\infty u(\lambda \varepsilon n) p_{\varepsilon n} \leq u(-\lambda\varepsilon) \left(F_{-\varepsilon} + \frac{u(\lambda\varepsilon)}{u(-\lambda\varepsilon)}\sum_{n =1}^\infty n p_{\varepsilon n} \right) \\
&\leq u(-\lambda\varepsilon) \left(F_{-\varepsilon} + \frac{u(\lambda\varepsilon)}{u(-\lambda\varepsilon)} \left(\frac{1}{\varepsilon} \mathbb{E}[X^+] + 1\right)  \right).
\end{align*}
Since $\lim_{\lambda\to\infty}\frac{u(-\lambda\varepsilon)}{u(\lambda\varepsilon)}=-\infty$, it follows that $u(-\lambda\varepsilon) < 0$ and $|\frac{u(\lambda\varepsilon)}{u(-\lambda\varepsilon)}| (\frac{1}{\varepsilon} \E[X^+]+1)  < F_{-\varepsilon}$
for all $\lambda >0$ sufficiently large, and hence
$\cE_u(\lambda X) < 0$ for all $\lambda >0$ sufficiently large.
\hfill\qed

\medskip

{\bf Proof of Proposition~\ref{prop:Cu SLL}}. The properties of $u$ imply that $C_u(X) \geq 0$ if and only if $\cE_u(X) \geq 0$ for each $X \in L^1$. Hence $C_u$ is sensitive to large losses if and only if $\cE_u$ is. Now the result follows from Theorem \ref{thm:expect utility}.\hfill\qed

\medskip

{\bf Proof of Proposition~\ref{prop:Mu SLL}}.
“(a) $\Rightarrow$ (b)”: This is trivial. “(b) $\Rightarrow$ (c)”: We argue by contraposition. To that end, suppose that $\limsup_{x \to \infty} \tfrac{u(-x)}{u(x)} > -\infty$. Then, by Theorem \ref{thm:expect utility}, there exists $X \in L^\infty$ and a sequence $(\lambda_n) \subset \R$ with $\mathbb{P}(X<0)>0$ and $\lambda_n \to \infty$ such that $\mathbb{E}[u(\lambda _n X)] \geq 0$ for every $n\in\N$.  It follows by definition that $\cM_u(\lambda _n X) \geq 0$ for every $n\in\N$, showing that $\cM_u$ is not sensitive to large losses on $L^\infty$.

“(c) $\Rightarrow$ (a)”: Assume $\limsup_{x \to \infty} \tfrac{u(-x)}{u(x)} = -\infty$. Let $\tilde{u}$ be the right-continuous modification of $u$. Since $u$ is increasing, it has at most countably many discontinuities. It follows that $\tilde{u}$ is increasing, negatively star-shaped on $\mathbb{R}_+$ and satisfies $\tilde{u}(x) < 0$ for every $x \in (-\infty,0)$ as well as $\limsup_{x \to \infty} \tfrac{\tilde{u}(-x)}{\tilde{u}(x)} = -\infty$. Note $\tilde{u}(0)$ may be nonzero, however $M_{\tilde{u}}$ is a utility functional. Since in addition $\tilde{u} \geq u$ and therefore $M_{\tilde u} \geq M_u$, it suffices to show the stronger claim that $M_{\tilde u}$ is sensitive to large losses on $L^1$. To establish the latter, it suffices to show that $\cE_{\tilde{u}}:L^1 \to [-\infty, \infty)$ defined by $\cE_{\tilde{u}}(X):=\mathbb{E}[\tilde{u}(X)]$ (which may not be normalised) is upper regular. Indeed, both Theorem \ref{the: sensitivity via acceptability v2}[(a) $\Rightarrow$ (b)] and Theorem \ref{thm:expect utility}[(c) $\Rightarrow$ (a)] do not rely on normalisation and can be applied to $\cE_{\tilde{u}}$. 

So let $X \in L^1$, $(m_n) \subset \R$ and $m \in \R$ with $m_n \to m$. We must show 
\begin{equation}
\label{eq:Mu SLL:proof goal}
    \mathbb{E}[\tilde{u}(X+m)] \geq \limsup_{n \to \infty} \mathbb{E}[\tilde{u}(X+m_n)].
\end{equation}
By the monotonicity of $\tilde{u}$, we may assume without loss of generality that $m_n \searrow m$. Note that by right-continuity of $\tilde{u}$, $\tilde{u}(X+m_n) \to \tilde{u}(X+m) \ \mathbb{P}$-a.s. There are two subcases: If $\mathbb{E}[\tilde{u}(X+m)] \in \mathbb{R}$, then $\lvert \tilde{u}(X+m_n) \rvert \leq  \left( \tilde{u}(X+m)^- +  \tilde{u}(X+m_1)^+ \right) \in L^1$ for every $n$ and the claim follows by dominated convergence. If $\mathbb{E}[\tilde{u}(X+m)] = -\infty$, then by the asymptotic property \eqref{eq:def:utility fn:not superlinear} of $\tilde{u}$, $\mathbb{E}[\tilde{u}(X+m)^-] = \infty$. Since $\tilde{u}(X+m_n) \to \tilde{u}(X+m) \ \mathbb{P}$-a.s., and $m_n \searrow m$, $\tilde{u}(X+m_n)^- \nearrow \tilde{u}(X+m)^- \ \mathbb{P}$-a.s. The monotone convergence theorem implies $\lim_{n \to \infty} \mathbb{E}[\tilde{u}(X+m_n)^-] = \infty$. On the other hand, $\mathbb{E}[\tilde{u}(X+m_n)^+] \leq \mathbb{E}[\tilde{u}(X+m_1)^+] < \infty$. Therefore, $\lim_{n \to \infty} \mathbb{E}[\tilde{u}(X+m_n)] = -\infty$ and \eqref{eq:Mu SLL:proof goal} also holds in this case.\hfill\qed

\medskip

\begin{lemma}
\label{lem:OCE concave utility}
Let $u$ be a concave utility function such that $u(x) \leq x$ for $x \in \R$. Assume that $\lim_{x\to-\infty}\tfrac{u(x)}{x} = \infty$ and $\lim_{x\to\infty}\tfrac{u(x)}{x}=0$. Then $\OCE_u$ is sensitive to large losses on $L^1$.

\end{lemma}

\begin{proof}
We argue via contradiction. So suppose there exists $X \in L^1$ with $\mathbb{P}(X < 0) > 0$ such that
\begin{equation}
\label{eq:contradicting assumption OCE}
    \OCE_u(\lambda X) \geq 0, \quad \lambda \in (0,\infty).
\end{equation}
By monotonicity of $\OCE_u$, we may assume without loss of generality that $\essinf X \in (-\infty, 0)$. We will show that there exists $\overline{\lambda}_X, \overline{\eta}_X \in (0, \infty)$ such that 
\begin{equation}
\label{eq:OCE goal}
    \OCE_u(\lambda X) = \sup_{\eta \in [0, \overline{\eta}_X]} \{ \eta + \mathbb{E}[u(\lambda X - \eta)] \}, \quad \lambda \in (\overline{\lambda}_X, \infty).
\end{equation}
If we can show this, then $\OCE_u(\lambda X) \leq \overline{\eta}_X + \mathbb{E}[u(\lambda X)]$ for $\lambda \in (\overline{\lambda}_X, \infty)$. But since $\mathbb{P}(X<0)>0$, it follows from the proof of Theorem \ref{thm:expect utility}[(c) $\Rightarrow$ (a)] that there exists $\tilde{\lambda}_X \in (0, \infty)$ such that $\mathbb{E}[u(\lambda X)] < -\overline{\eta}_X$ for $\lambda \in (\tilde{\lambda}_X, \infty)$. Whence, $\OCE_u(\lambda X) < 0$ for $\lambda \in (\overline{\lambda}_X \vee \tilde{\lambda}_X, \infty)$, in contradiction to \eqref{eq:contradicting assumption OCE}. Therefore, it suffices to show \eqref{eq:OCE goal} holds under \eqref{eq:contradicting assumption OCE}.

\medskip
Since $\mathbb{P}(X < 0) > 0$, there exists $\varepsilon > 0$ such that the event $A_\varepsilon:=\{ X \leq -\varepsilon \}$ has nonzero probability. For $\eta \geq 0$ and $\lambda > 0$, we have
\begin{align}
    \eta + \mathbb{E}[u(\lambda X - \eta)] & \leq \eta + \mathbb{E}[u(\lambda X - \eta)\ind_{A_{\varepsilon}}] + \mathbb{E}[u(\lambda X - \eta)\ind_{\{X \geq 0\}}] \nonumber \\ & \leq \eta + u(-\lambda \varepsilon - \eta)\mathbb{P}(A_\varepsilon) + \mathbb{E}[u(\lambda X^+)] \label{eq:eta nonnegative} \\ & \leq \eta + u(-\eta)\mathbb{P}(A_\varepsilon) + u(-\lambda \varepsilon)\mathbb{P}(A_\varepsilon) + \lambda \mathbb{E}[X^+],\nonumber
\end{align}
where the last inequality follows from the facts $u$ is concave so $u(-\lambda \varepsilon - \eta) \leq u(-\lambda \varepsilon) + u(-\eta)$, and $u(x) \leq x$. Because $\lim_{x \to -\infty} u(x)/x = \infty$, that means $\lim_{\lambda \to \infty} (u(-\lambda \varepsilon)\mathbb{P}(A_{\varepsilon}))/(\lambda \mathbb{E}[X^+]) = -\infty$, i.e., there exists $\lambda_1 \in (0, \infty)$ such that $u(-\lambda \varepsilon)\mathbb{P}(A_\varepsilon) + \lambda \mathbb{E}[X^+] <0$ for every $\lambda \in (\lambda_1, \infty)$. Similarly, there exists $\overline{\eta}_X \in (0,\infty)$ such that $\eta + u(-\eta) \mathbb{P}(A_\varepsilon) < 0$ for every $\eta \in (\overline{\eta}_X,\infty)$. It follows from \eqref{eq:eta nonnegative} that for $\lambda \in (\lambda_1,\infty)$ and $\eta \in (\overline{\eta}_X,\infty)$, 
\begin{equation*}
    \eta + \mathbb{E}[u(\lambda X - \eta)] < 0.
\end{equation*}
Together with \eqref{eq:contradicting assumption OCE}, this implies
\begin{equation}
\label{eq:lambda 1}
    \OCE_u(\lambda X) = \sup_{\eta \in (-\infty, \overline{\eta}_X]} \{ \eta + \mathbb{E}[u(\lambda X - \eta)] \}, \quad  \lambda \in (\lambda_1, \infty).
\end{equation}
For $\lambda > 0$ and $\eta \leq \lambda \essinf X$, we have $\lambda X - \eta \geq 0$. Since $u$ is concave and $u(x) \leq x$ for $x \in \R_+$,
\begin{equation*}
    \eta + \mathbb{E}[u(\lambda X - \eta)] \leq \mathbb{E}[u(\lambda X)].
\end{equation*}
Since $\mathbb{P}(X < 0) > 0$, by Theorem \ref{thm:expect utility}, there exists $\lambda_2 \in (0,\infty)$ such that $\mathbb{E}[u(\lambda X)] < 0$ for every $\lambda \in (\lambda_2, \infty)$. This, together with \eqref{eq:lambda 1} and \eqref{eq:contradicting assumption OCE}, implies that
\begin{equation}
\label{eq:max of lambda 1 and lambda 2}
    \OCE_u(\lambda X) = \sup_{\eta \in (\lambda \essinf X, \overline{\eta}_X]} \{ \eta + \mathbb{E}[u(\lambda X - \eta)] \}, \quad  \lambda \in (\lambda_1 \vee \lambda_2, \infty).
\end{equation}
Now fix $\delta \in (0, \essinf X)$. Since $\lim_{x\to\infty}u(x)/x = 0$, for every $a \in (0,\infty)$, there exists $M_a \in (0,\infty)$ such that $u(x) \leq ax$ for $x \in (M_a, \infty)$. Choose $a$ small enough so that $K:=(\essinf X + \delta) + a(\mathbb{E}[X^+] - \essinf X) < 0$ and let $B_{M_{a}}:=\{ X \leq M_a \}$. Then for $\lambda > 0$ and $\lambda \essinf X < \eta \leq \lambda (\essinf X + \delta)$, we have
\begin{align*}
    \eta + \mathbb{E}[u(\lambda X - \eta)] & \leq \lambda (\essinf X + \delta) + \mathbb{E}[u(\lambda (X - \essinf X))] \\ & = \lambda (\essinf X + \delta) + \mathbb{E}[u(\lambda (X - \essinf X))\ind_{B_{M_{a}}}] + \mathbb{E}[u(\lambda (X - \essinf X))\ind_{B_{M_{a}}^c}] \\ & \leq \lambda (\essinf X + \delta) + u(\lambda(M_a-\essinf X)) + \lambda a(\mathbb{E}[X^+] - \essinf X) \\ &
    = \lambda K + u(\lambda(M_a-\essinf X)).
\end{align*}
It follows from $\lim_{x\to\infty}u(x)/x = 0$ and $K <0$ that there exists $\lambda_3 \in (0,\infty)$ such that $\lambda K + u(\lambda(M_a-\essinf X)) < 0$ for every $\lambda \in (\lambda_3, \infty)$. Together with \eqref{eq:max of lambda 1 and lambda 2} and \eqref{eq:contradicting assumption OCE}, this implies that
\begin{equation}
\label{eq:max of lambda 1 and lambda 2 and lambda 3}
    \OCE_u(\lambda X) = \sup_{\eta \in (\lambda (\essinf X+\delta), \overline{\eta}_X]} \{ \eta + \mathbb{E}[u(\lambda X - \eta)] \}, \quad  \lambda \in (\lambda_1 \vee \lambda_2 \vee \lambda_3, \infty).
\end{equation}
Finally, for $\lambda > 0$ and $\eta \in (\lambda (\essinf X + \delta), 0]$, 
\begin{equation*}
    \eta + \mathbb{E}[u(\lambda X - \eta)] \leq \eta + \mathbb{E}[u(\lambda(X-(\essinf X + \delta))].
\end{equation*}
Since $\mathbb{P}((X-(\essinf X + \delta) < 0) > 0$, by Theorem \ref{thm:expect utility}, there exists $\lambda_4 \in (0,\infty)$ such that $\mathbb{E}[u(\lambda(X-(\essinf X + \delta))] < 0$ for every $\lambda \in (\lambda_4, \infty)$. This, together with \eqref{eq:max of lambda 1 and lambda 2 and lambda 3} and \eqref{eq:contradicting assumption OCE}, implies
\begin{equation*}
    \OCE_u(\lambda X) = \sup_{\eta \in [0, \overline{\eta}_X]} \{ \eta + \mathbb{E}[u(\lambda X - \eta)] \}, \quad  \lambda \in (\overline{\lambda}_X, \infty),
\end{equation*}
where $\overline{\lambda}_X:=\max(\lambda_1, \lambda_2, \lambda_3, \lambda_4)$, as desired.
\end{proof}

\medskip

{\bf Proof of Theorem \ref{theo: OCE}}. “(a) $\Rightarrow$ (b)”: This is trivial.

“(b) $\Rightarrow$ (c)”: We argue by contraposition. First, assume that $\liminf_{x \to -\infty} \tfrac{u(x)}{x} = a < \infty$. Then there exists a sequence $(\lambda_n)_{n \in \N}$ with $\lim_{n \to\infty} \lambda_n = \infty$ such that $u(-2 \lambda_n) \geq - 2\lambda_n (a+1)$ for all $n \in \NN$. Choose $A \in \cF$ with $0 < \mathbb{P}(A) \leq \frac{1}{2 (a+1)}$ and set $X := -\ind_A + \ind_{A^c}  \in L^\infty$. Then $\mathbb{P}(X < 0) > 0$ and $X -1 = -2 \ind_A$. Moreover, for each $n \in \NN$,
\begin{equation*}
\OCE_u(\lambda_n X) \geq \lambda_n + \mathbb{E}[u(\lambda_n (X - 1))] =\lambda_n + \probp(A) u(-2 \lambda_n)\geq 0.
\end{equation*}
Thus, $\OCE_u$ is not sensitive to large losses on $L^\infty$.

Next, assume that $\limsup_{x\to+\infty}\tfrac{u(x)}{x} =b > 0$ and note that $b \leq 1$ by the fact that $u(x) \leq x$ for $x \in \R$. Then for every $\kappa \in (0,1)$ and $M \in (0,\infty)$, there exists a sequence $(\lambda_n)_{n \in \NN}$ with $\lim_{n \to \infty} \lambda_n = \infty$ such that $u(\lambda_n M) \geq \lambda_n M \kappa b$. Choose $\kappa = \frac{3}{4}$ and $M = \frac{2}{b^2}$ such that $u(\lambda_n \frac{2}{b^2} ) \geq \lambda_n \frac{3 }{2 b}$ for all $n \in \NN$. Choose $B \in \cF$ with $1 > \mathbb{P}(B) \geq \frac{2}{3} b$ and set $X := -\ind_{B^c} + (\frac{2}{b^2} - 1)\ind_{B} \in L^\infty$. Then $\mathbb{P}(X < 0) > 0$ and $X + 1 = \frac{2}{b^2} \ind_{B} $. Moreover, for each $n \in \N$,
\begin{equation*}
\OCE_u(\lambda_n X) \geq -\lambda_n + \mathbb{E}[u(\lambda_n (X + 1))] = -\lambda_n + \probp(B) u\left(\lambda_n \frac{2}{b^2}\right)  \geq 0.
\end{equation*}
Thus, $\OCE_u$ is not sensitive to large losses on $L^\infty$.

“(c) $\Rightarrow$ (a)”: Assume $\liminf_{x\to-\infty}\tfrac{u(x)}{x} = \infty$ and $\limsup_{x\to\infty}\tfrac{u(x)}{x}=0$. Let $\tilde{u}$ be the concave hull of $u$, i.e., the smallest concave function that majorizes $u$ in a pointwise way. The assumptions on $u$ imply that $\tilde{u}(x) \leq x$ for $x \in \R$ and that $\lim_{x\to-\infty}\tfrac{\tilde{u}(x)}{x} = \infty$ as well as $\lim_{x\to\infty}\tfrac{\tilde{u}(x)}{x}=0$. Thus, $\OCE^{\tilde{u}}$ is sensitive to large losses on $L^1$ by Lemma \ref{lem:OCE concave utility}. Since $\tilde u \geq u$, it follows that $\OCE^{\tilde{u}} \geq \OCE_u$ and therefore $\OCE_u$ is a fortiori sensitive to large losses on $L^1$.\hfill\qed


\end{document}